\documentclass[11pt]{article} 

\usepackage[utf8]{inputenc} 

\usepackage[margin=1in]{geometry}
\usepackage{graphicx} 

\usepackage{booktabs} 
\usepackage{array} 
\usepackage{paralist} 
\usepackage{verbatim} 
\usepackage{subfig} 
\usepackage{transparent}
\usepackage{color}
\usepackage{fancyhdr} 
\usepackage{sectsty}
\allsectionsfont{\sffamily\mdseries\upshape} 

\usepackage[nottoc,notlof,notlot]{tocbibind} 
\usepackage[titles,subfigure]{tocloft} 

\usepackage{amsmath,amsthm}
\usepackage{amssymb}
\usepackage{slashed}
\usepackage{accents}
\usepackage{wrapfig}
\newcommand{\ubar}[1]{\underaccent{\bar}{#1}}
\newcommand{\s}{\slashed}

\newtheorem*{theorem*}{Theorem}
\newtheorem{theorem}{Theorem}[section]
\newtheorem{proposition}[theorem]
{Proposition}
\newtheorem{lemma}{Lemma}[subsection]

\newtheorem{definition}{Definition}[section]
\newtheorem{remark}{Remark}[section]

\DeclareMathOperator{\tr}{tr}
\DeclareMathOperator{\II}{II}

\title {Quasi-Round MOTSs and Stability of the Schwarzschild\\ Null Penrose Inequality}
\author{Henri Roesch}
\date{\today} 

\begin{document}
\maketitle

\begin{abstract}
In \cite{R}, the notion of Double Convexity for a foliation of a conical null hypersurface was introduced to give a proof, if satisfied, of the Null Penrose Inequality. Double Convexity constrains the geometry of a Marginally Outer Trapped Surface (MOTS), called a quasi-round MOTS. In the first part of this paper, for a class of strictly stable Weakly Isolated Horizons, we show the existence of a unique foliation by quasi-round MOTS. In the second part, we show that any subsequent space-time perturbation continues to admit a quasi-round MOTS. Finally, for perturbations of the quasi-round MOTS in Schwarzschild, we identify sufficient conditions on the asymptotics of any past-pointing null hypersurface that yields the Null Penrose Inequality.
\end{abstract}

\tableofcontents

\section{Introduction}
In \cite{P1,P2}, Roger Penrose conjectured a geometric inequality to represent the physically natural notion, that the total mass of an isolated physical system should be no larger than the combined mass of its black holes. The nuance in this simple statement comes form the fact that the total mass, or the \textit{ADM Mass}, $M_{ADM}$ (see \cite{arnowitt2008republication}), is a geometric invariant measured `at infinity', and the black hole, $\Sigma_H$, has a quasi-local mass involving the area of it's boundary. Denoting the black hole area $|\Sigma_H|$, the Penrose Conjecture takes the form:
$$\sqrt{\frac{|\Sigma_H|}{16\pi}}\leq M_{ADM}.$$
Following a heuristic argument, Penrose formulated this inequality from various ingredients, some mathematically verified, but all physically reasonable, as a viability test of the \textit{weak Cosmic Censorship hypothesis}. This hypothesis provides a remedy to the development of singularities as shown in the famous singularity theorems of Hawking \cite{Hawking:1965mf}, and Penrose \cite{P3}, by positing a fundamental structure for solutions to Einstein's field equations that hide these singularities, namely black holes. If the Penrose Conjecture is shown to be false, this would deal a deadly blow to Cosmic Censorship. Cases in which the conjecture have been verified include spherical symmetry \cite{Ha, PhysRevD.49.6931}, and time symmetric slices \cite{huisken2001inverse},\cite{bray2001proof, bray2009riemannian}. The setting of this paper involves a more recent approach utilizing a null conical hypersurface or \textit{Null Cone} (Definition \ref{d8}).\\
\indent From this perspective, a quasi-local black hole is represented by a spacelike 2-sphere, $\Sigma_0$, with vanishing future null expansion, called a \textit{Marginally Outer Trapped Surface} (or MOTS). The collection, or congruence, of ingoing light rays that orthogonally intersect $\Sigma_0$ rule (at least in a neighborhood of $\Sigma_0$) a smooth null hypersurface. We will refer to a \textit{Null Cone}, $\Omega$, whenever this collection of light rays continue to extend to the infinite past as a smooth conical null hypersurface (see Definition \ref{d8}). Arguably the simplest and best known example of a Null Cone is given by the standard light-cone, $\Lambda$, of flat Minkowski space. Alternatively, one may consider the famous Schwarzschild spacetime modeling a static spherically symmetric black hole in vacuum. Taking $\Sigma_0$ to be a standard spherically symmetric MOTS yields a Null Cone, $\Omega_S$, indistinguishable from the point of view of its intrinsic geometry to that of the light-cone $\Lambda$. We will simply refer to this spherically symmetric null hypersurface, $\Omega_S$, as the \textit{standard Schwarzschild Null Cone}. For general Null Cone geometries, one obtains a foliation (or a flow within $\Omega$), by spherical leaves that expand pointwise along these rays to past `null infinity'. For every flow whereby the evolving family of induced metrics asymptotically rescale to the standard round metric on $\mathbb{S}^2$, one identifies an abstract `observer at infinity'. One also identifies a fixed velocity $\vec{v}$ (with $|\vec{v}|<1$, since the speed of light, $c$, is scaled to $c=1$) relative to the isolated system, and an associated total \textit{Trautman-Bondi energy}, $E_{TB}(\vec{v})$. By taking an infimum, we obtain the \textit{Trautman-Bondi mass}, $m_{TB}:=\inf_{\{\vec{v}||\vec{v}<1\}}E_{TB}(\vec{v})$ (see Definition \ref{d15}). The Penrose Conjecture in this setting takes the form,
$$\sqrt{\frac{|\Sigma_0|}{16\pi}}\leq m_{TB}(\Omega).$$
This form of the conjecture is often referred to as the \textit{Null Penrose Inequality} and has been verified when $\Omega$ is shear-free in a vacuum spacetime by J. Sauter \cite{S}. For small vacuum perturbations of the metric around the spherically symmetric Null Cone in the Schwarzschild spacetime, denoted $\Omega_S$, the inequality has been announced to hold for the \textit{Weak Null Penrose Conjecture}, namely for the weaker upper bound $E_{TB}(\vec{v})$, by S. Alexakis \cite{A}. One also obtains a geometric analogue of the Null Penrose Inequality for surfaces in the flat Minkowski spacetime, the so called \textit{Gibbons-Penrose Conjecture}, resulting from Penrose's original consideration of a collapsing thin shell of null dust, \cite{Gibbons_1997}. Interestingly, the Gibbons-Penrose Conjecture generalizes the classical Minkowski inequality for convex 2-spheres in $\mathbb{R}^3$. Work by M.T. Wang \cite{MT1}, and M. Mars-A. Soria \cite{0264-9381-29-13-135005} verified the Gibbons-Penrose conjecture for a collection of surfaces in Minkowski spacetime, although the general case remains open. An interesting analogue of the thin-shell case is also verified for a variety of surfaces in the Schwarzschild spacetime by S. Brendle-M.T. Wang \cite{Brendle2014}.\\
\indent A general proof of the Weak Null Penrose Inequality was claimed by M. Ludvigsen and J.A.G Vickers \cite{ludvigsen1983inequality}, but G. Bergqvist \cite{bergqvist1997penrose} pointed out that no guarantee of `asymptotic roundness' for their given foliation had been justified in order to secure comparison with total energy. In \cite{S}, Sauter also runs into this difficulty. Sauter identified two different flows, namely a `uniformly area expanding flow', and a `constant mass aspect flow', both yielding a non-decreasing Hawking Energy (see Definition \ref{d14}). Subsequently, for small metric perturbations around $\Omega_S$ in the Schwarzschild spacetime, Sauter was able to bound the MOTS mass $\sqrt{|\Sigma_0|/16\pi}$ by the asymptotic limit of Hawking Energy. Unfortunately, the needed asymptotic roundness does not necessarily develop. In \cite{MS2}, M. Mars - A. Soria showed, under fairly generic assumptions, the existence of a foliation of $\Omega$ called \textit{Geodesic Asymptotically Bondi} (GAB), that exhibits the necessary decay in the Ludvigsen-Vickers-Bergqvist approach. With the use of a new energy functional, these authors were also able to bound the MOTS mass by the asymptotic limit of the Hawking Energy. Again, the GAB foliation does not necessarily become round at infinity, and therefore the difficulty of relating the resulting limit to total energy persists.\\ 
\indent In \cite{A}, Alexakis employs a careful analysis of the induced geometry of spherical cuts at null infinity that result from the uniformly area expanding flow, or \textit{luminosity flow}, in Sauter's work. Using an Implicit Function Theorem argument `at null infinity', Alexakis was able to choose a Null Cone exhibiting the desired asymptotics inside small vacuum perturbations of Schwarzschild spacetime. In this paper, we approach this same problem by showing that one can instead use an Implicit Function Theorem to choose a viable MOTS within small metric perturbations, putting us in a position to apply recent results in \cite{R}. This allows us to prove the full Null Penrose Inequality in this setting, even for non-vacuum perturbations with `reasonable decay' (see (\ref{e18}-\ref{e21}) below).\\
\indent More specifically, in \cite{R} the author constructs a new mass functional $m(\Sigma)$ for a spacelike 2-sphere $\Sigma$. Given convexity conditions on $\Sigma$ (see Definition \ref{d5} below), it follows (see Proposition \ref{p1} below) that $\frac{d}{dt}|_{t=0}m(\Sigma_t)\geq 0$ along -any- past-pointing null flow $\{\Sigma_t\}$ off of $\Sigma$. Consequently, if $\{\Sigma_t\}$ is also a foliation of $\Omega$ where each leaf $\Sigma_t$ respects these convexity conditions, called a \textit{doubly convex} foliation, we deduce a non-decreasing mass, since then $\frac{d}{dt}m(\Sigma_t)\geq 0$ for all $t$. In \cite{R}, it is also shown that the mass converges and the limiting value is independent of any choice of foliation in a neighborhood of infinity. Moreover, simply the existence of a doubly convex foliation enforces that this limit underestimates $m_{TB}$. Now, as a consequence of the Strong Maximum Principle whenever these convexity conditions are satisfied on a MOTS, called a \textit{quasi-round MOTS}, one also observes the desired `initial' mass $m(\Sigma_0)=\sqrt{|\Sigma_0|/16\pi}$. Therefore, in the final part of this work, we start by applying earlier results to verify the existence of a quasi-round MOTS for a small metric perturbation of the Schwarzschild geometry around a quasi-round MOTS within $\Omega_S$. In the $\Omega_S$ case, the spherically symmetric leaves yield a doubly convex foliation off of a quasi-round (in-fact perfectly round) MOTS. Moreover, the radial coordinate $r$ parametrizing these leaves is affine. In Section 4 (see Propositions \ref{p7}, and \ref{p9}), we show sufficiently small perturbations of $\Omega_S$ continue to support doubly convex affine foliations off our quasi-round MOTS, provided the asymptotic decay conditions (\ref{e18}-\ref{e21}) are satisfied. Consequently:
 $$\sqrt{\frac{|\Sigma_0|}{16\pi}}=m(\Sigma_0)\leq \lim_{s\to\infty}m(\Sigma_s)\leq m_{TB}.$$
 \subsection{Initial Constructions and known Results}
 In this section we start by un-packing the details of a main result in \cite{R} (see Theorem 1.1) that we will need in this paper. We will largely borrow from the construction presented there necessary for describing the main results of this paper.\\
\indent A spacetime $(\mathcal{M},g)$ is defined to be a four dimensional smooth manifold $\mathcal{M}$ equipped with a metric $g(\cdot,\cdot)$ (or $\langle\cdot,\cdot\rangle$) of Lorentzian signature $(-,+,+,+)$. We assume that the spacetime is both orientable and time orientable, i.e. admits a nowhere vanishing timelike vector field, defined to be future-pointing. Our convention for the Riemann curvature tensor is given by:
$$R_{XY}Z = D_{[X,Y]}Z-[D_X,D_Y]Z,$$
where $X,Y,Z\in\Gamma(T\mathcal{M})$. For a local orthonormal frame $\{e_i\}\subset \Gamma(TU)$, $U\subset \mathcal{M}$, such that $\langle e_0,e_0\rangle = -1$, $\langle e_j,e_j\rangle = 1$, $1\leq j\leq 3$, we define the Ricci tensor, $\text{Ric}$, and scalar, $R$:
\begin{align*}
\text{Ric}(X,Y) &= -\langle R_{e_0 X}e_0,Y\rangle+\sum_{j=1}^3\langle R_{e_j X}e_j,Y\rangle,\\
R &= -\text{Ric}(e_0,e_0)+\sum_{j=1}^3\text{Ric}(e_j,e_j).
\end{align*}
\indent We will also assume the \textit{Dominant Energy Condition} holds in $(\mathcal{M},g)$, namely that the ambient Einstein Curvature Tensor $G:=\text{Ric}-\frac12 R g$ satisfies $G(T,S)\geq 0$ for $S,T\in\Gamma(T\mathcal{M})$ future pointing timelike vector fields.\\
\indent Throughout this paper, we will denote a spacelike embedding of a sphere $\iota:\mathbb{S}^2\hookrightarrow\mathcal{M}$ by $\Sigma:=\iota(\mathbb{S}^2)$, with induced (Riemannian) metric $\gamma$. We will also denote the set of smooth functions on $\Sigma$ by $\mathcal{F}(\Sigma)$. It is well known that $\Sigma$ has trivial normal bundle, $T^{\perp}\Sigma$, with induced metric of signature $(-,+)$. Therefore, from any choice of null section $\ubar L\in \Gamma(T^{\perp}\Sigma)$, we have a unique null partner $L\in\Gamma(T^{\perp}\Sigma)$ satisfying $\langle \ubar L,L\rangle = 2$. This provides $T^{\perp}\Sigma$ with a null basis $\{L,\ubar L\}$. 

\begin{figure}[h]
\centering
\includegraphics[width=0.75\textwidth]{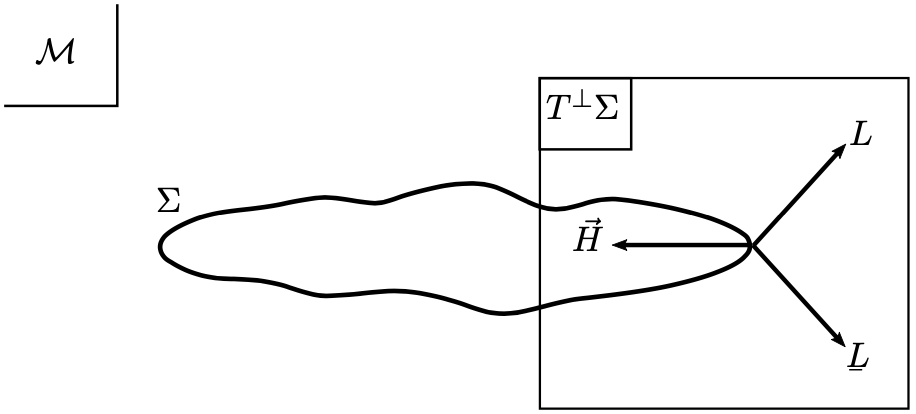}
\end{figure}

We also notice that any ``boost" $\{\ubar L,L\}\to\{\ubar L_a,L_a\}$ given by:
$$\ubar L_a:= a\ubar L,\,\,\, L_a:= \frac{1}{a}L$$ 
(for $a\in\mathcal{F}(\Sigma)$ a non-vanishing smooth function on $\Sigma$) insures $\langle \ubar L_a,L_a\rangle=\langle \ubar L,L\rangle=2$ so we observe a gauge freedom in our choice of basis for $T^\perp\Sigma$.\\
\indent Our convention for the second fundamental form II and mean curvature $\vec{H}$ of $\Sigma$ are
$$\text{II}(V,W) = D^{\perp}_VW,\,\,\,\,\vec{H} = \tr_\Sigma\text{II}$$
for $V,W\in \Gamma(T\Sigma)$, where $D$ represents the Levi-Civita connection of the ambient spacetime. 
\begin{definition} 
\label{d1}
Given a choice of null basis $\{\ubar L,L\}$, following the conventions of Sauter \cite{S}, we define the associated symmetric 2-tensors $\ubar\chi,\chi$ and torsion (connection 1-form), $\zeta$, by
\begin{align*}
\ubar\chi(V,W) &:= \langle D_V\ubar L,W\rangle = -\langle \ubar L,\II(V,W)\rangle\\
\chi(V,W) &:=\langle D_VL,W\rangle = -\langle L,\II(V,W)\rangle\\
\zeta(V) &:= \frac12\langle D_V\ubar L,L\rangle=-\frac12\langle D_VL,\ubar L\rangle
\end{align*}
where $V,W\in \Gamma(T\Sigma)$.
\end{definition}
Any boosted basis $\{\ubar L_a,L_a\}$ produces the associated tensors of Definition \ref{d1}:
\begin{align*}
\ubar\chi_a(V,W)&:=\langle D_V(a\ubar L),W\rangle = a\ubar\chi(V,W)\\
\chi_a(V,W)&:=\langle D_V(\frac{1}{a}L),W\rangle = \frac{1}{a}\chi(V,W)\\
\zeta_a(V)&:=\frac12\langle D_V(a\ubar L),\frac{1}{a}L\rangle = \zeta(V)+V\log|a|=(\zeta+{d}\log|a|)(V).
\end{align*}
For a symmetric 2-tensor $T$ on $\Sigma$, its \textit{trace-free} (or \textit{trace-less}) part is given by
$$\hat{T}:=T-\frac12(\tr_\gamma T)\gamma$$
allowing us to decompose $\ubar\chi$ into its \textit{shear} and \textit{expansion} components respectively:
$$\ubar\chi = \hat{\ubar\chi}+\frac12(\tr\ubar\chi)\gamma.$$
\begin{definition}
\label{d2}
We say $\Sigma$ \textit{is expanding along $\ubar L$} for some null section ${\ubar L\in\Gamma(T^{\perp}\Sigma)}$ provided that,
\begin{equation}
\langle-\vec{H},\ubar L\rangle = \tr\ubar\chi>0\tag{\dag}
\end{equation}
on all of $\Sigma$. 
\end{definition} 
Any infinitesimal flow of $\Sigma$ along $\ubar L$ gives, by first variation of area, $\delta_{\ubar L}{dA} = \langle-\vec{H},\ubar L\rangle dA= \tr\ubar\chi dA$. So the flow is locally area expanding, $\delta_{\ubar L}{dA}>0$, only if $\Sigma$ ``is expanding along $\ubar L$". For $\Sigma$ expanding along some $\ubar L\in\Gamma(T^\perp\Sigma)$ we are able to choose a canonical null basis $\{L^-,L^+\}$ by requiring that our flow along $L^- = a\ubar L$ be uniformly area expanding, $\delta_{L^-}{dA} = dA$. From first variation of area, flowing along $a\ubar L$ gives
$$\delta_{a\ubar L}{dA} = -\langle\vec{H},a\ubar L\rangle dA = a\tr\ubar\chi dA.$$
So we achieve a uniformly area expanding null flow along $L^- = \frac{\ubar L}{\tr\ubar\chi}$.
\begin{definition}\label{d3} 
For $\Sigma$ expanding along some $\ubar L\in\Gamma(T^\perp\Sigma)$, we call the associated uniformly area expanding null basis $\{L^-,L^+\}$ given by
$$L^-:=\frac{\ubar L}{\tr\ubar\chi},\,\,\, L^+:=\tr\ubar\chi L$$
the \textit{null inflation basis}.
\end{definition}
\indent We also define $\chi^{-(+)} := -\langle \II,L^{-(+)}\rangle$. We observe that 
\begin{align*}
\tr\chi^- &= 1\\ 
\tr\chi^+ &= \tr\ubar\chi\tr\chi = \langle\vec{H},\vec{H}\rangle,
\end{align*}
and for $V\in\Gamma(T\Sigma)$, the torsion associated to this basis is given by
$$\tau(V)=\frac12\langle D_VL^-,L^+\rangle = (\zeta-{d}\log\tr\ubar\chi)(V).$$
\begin{remark}
The vector field $L^-$ is precisely the ``velocity" that generates the luminosity flow considered in the work of both Sauter \cite{S}, and Alexakis \cite{A}.
\end{remark}
We will denote the induced covariant derivative on $\Sigma$ by $\nabla$.
\begin{definition}\label{d4}
Assuming $\Sigma$ is expanding along $\ubar L\in\Gamma(T^\perp\Sigma)$, we define the geometric flux function
\begin{equation}
\rho = \mathcal{K}-\frac14\langle\vec{H},\vec{H}\rangle+ \nabla\cdot\tau \label{e1}
\end{equation}
where $\mathcal{K}$ represents the Gaussian curvature of $\Sigma$.\\
This allows us to define an associated quasi-local mass
\begin{equation}
m(\Sigma) = \frac12\Big(\int_\Sigma\rho^{\frac23}\frac{dA}{4\pi}\Big)^{\frac32} \label{e2}.
\end{equation}
\end{definition}
We note here that the flux function $\rho$ is related to the similarly denoted component of the Weyl curvature tensor in \cite{christodoulou2014global}. More specifically, $\rho$ agrees with the conjugate mass aspect function $\bar{\mu}$ of \cite{christodoulou2014global} under the gauge choice $\{L^-,L^+\}\subset \Gamma(T^\perp\Sigma)$. We refer the reader to \cite{R} for a full description behind the utility of the mass functional (\ref{e2}). Specific to the the Schwarzschild Null Cone $\Omega_S$ we refer the reader to Remark \ref{r2}. Roughly speaking, as we will also see in the proof of Theorem \ref{t4}, this mass functional is insensitive to `infinitesimal boosts' along a foliation of a Null Cone. As expected from a mass, this is unlike the behavior of energy functionals used in prior attempts at the Null Penrose Inequality. We will shortly observe that (\ref{e2}) also exhibits favorable monotonicity.\\
\indent For a normal null flow off of some $\Sigma$ with null flow vector $\ubar L$, technically the flow speed is zero since $\langle \ubar L,\ubar L\rangle = 0$. In case $\Sigma$ is expanding along $\ubar L$, we may define the \textit{\textit{expansion speed}}, $\sigma$, according to $\ubar L = \sigma L^-$. We notice that $\sigma = \tr\ubar\chi$.
We now state the first main result of \cite{R}:
\begin{proposition}[\cite{R}, Theorem 1.1]\label{p1}
Let $\Omega$ be a Null Cone foliated by spacelike spheres $\{\Sigma_s\}_s$ expanding along the null flow direction $\ubar L=\sigma L^-$ such that $|\rho(s)|>0$ for each $s$. Then the mass $m(s):=m(\Sigma_s)$ has rate of change 
\begin{align*}\frac{dm}{ds}=\frac{(2m)^\frac13}{8\pi}\int_{\Sigma_s}\frac{\sigma}{\rho^{\frac13}}\Big((|\hat{\chi}^-|^2+G(L^-,L^-))(\frac14\langle\vec{H},\vec{H}\rangle-\frac13\Delta\log|\rho|)+\frac12|\nu|^2+G(L^-,N)\Big)dA
\end{align*}
where,
$$\nu:=\frac23\hat{\chi}^-\cdot{d}\log|\rho|-\tau,\,\,\,
N:=\frac19|{d}\log|\rho||^2L^-+\frac13\nabla\log|\rho|-\frac14L^+.$$
\end{proposition}
We now observe that the section $N$ is in-fact a past-pointing null vector field. From the fact that the foliation $\{\Sigma_s\}_s$ is expanding along $\ubar L$ (i.e. $\sigma>0$) we consequently deduce from the Dominant Energy Condition that the mass functional $m(\Sigma_s)$ is monotonically increasing whenever the leaves of the foliation $\{\Sigma_s\}_s$ satisfies the following convexity conditions:
\begin{definition}\label{d5}
A spacelike 2-sphere $\Sigma$ is called doubly convex if it satisfies the conditions: 
\begin{align*}
\rho&>0\\
\frac14\langle\vec{H},\vec{H}\rangle&\geq\frac13\Delta\log\rho.
\end{align*}
\end{definition}
\indent We will assume for a Null Cone $\Omega$ (see Definition \ref{d8}) that $\ubar L$ is past pointing. In Section 2 we also show that a Null Cone $\Omega$ is characterized by the property that any cross-section $\Sigma\subset\Omega$ is expanding along any null (pre)geodesic generator $\ubar L\in\Gamma(T\Omega)$. As a result, we define a quasi-local black hole as follows:
\begin{definition}\label{d6}
Consider a spacelike 2-sphere $\Sigma\subset\mathcal{M}$ expanding along a null past-pointing section $\ubar L\in\Gamma(T^\perp\Sigma)$. We say $\Sigma$ is a Marginally Outer Trapped Surface (MOTS), if
$$\tr\chi = 0 \iff\langle\vec{H},\vec{H}\rangle = 0$$
on $\Sigma$.
\end{definition}
Here we observe, if a MOTS $\Sigma_0$ also satisfies the convexity conditions of Definition \ref{d5}, then
\begin{align*}
\mathcal{K}+\nabla\cdot\tau>0\geq \Delta\log(\mathcal{K}+\nabla\cdot\tau).\
\end{align*}
Applying the strong Maximum Principle for elliptic equations, the second inequality above constrains that $\mathcal{K}+\nabla\cdot\tau$ be constant on $\Sigma_0$. Applying the Gauss-Bonnet, and Divergence theorems, we therefore have
$$\rho = \mathcal{K}+\nabla\cdot\tau = \frac{4\pi}{|\Sigma|}.$$
It follows directly from (\ref{e2}), that $m(\Sigma_0) = \sqrt{\frac{|\Sigma_0|}{16\pi}}$, the black hole mass as considered by Roger Penrose.
\begin{definition}\label{d7}
Suppose we have a MOTS $\Sigma\subset\mathcal{M}$, we say $\Sigma$ is quasi-round whenever
$$\mathcal{K}+\nabla\cdot\tau=\frac{4\pi}{|\Sigma|}$$
on $\Sigma$.
\end{definition}
\subsection{Overview of Main Results}
In the Schwarzschild spacetime, the unique MOTS $\Sigma_0$ of $\Omega_S$ is not only quasi-round (in-fact round, see Section 2.1), it is also the 2-sphere of intersection, $\Sigma_0 = \Omega_S\cap\mathcal{H}_S$, with the Schwarzschild \textit{Killing Horizon}, $\mathcal{H}_S$. More generally, a Non Expanding Horizon (or NEH, Definition \ref{d9}) $\mathcal{H}$, similarly to a Null Cone, is a spherical congruence of light rays ruling a null hypersurface. In contrast to a Null Cone, a NEH is characterized by the property that all cross-sections $\Sigma\subset\mathcal{H}$ have vanishing future null expansion. Roughly speaking, any cross-section of $\mathcal{H}$ is a MOTS. Now a Killing Horizon such as $\mathcal{H}_S$, is a special sub-class of NEH with the property that any null tangent is proportional to the restriction of a Killing field from the ambient spacetime, say $\xi$. As we will see in Section 2, from the fact that $\xi|_{\mathcal{H}}=l$ generates null pregeodesics, we have $D_ll = \kappa l$. The fact that $\xi$ is also Killing forces $\kappa$, called the \textit{surface gravity}, to be a constant on all of $\mathcal{H}_S$. Ignoring all the properties of $\mathcal{H}_S$ other than constancy of surface gravity characterizes a broader class of NEH called \textit{Weakly Isolated Horizons} (or WIH, Definition \ref{d10}). In Section 2 we impose translation invariance for the past null expansion of leaves along a WIH, called \textit{optical rigidity} (Definition \ref{d11}). We subsequently show correlation between the sign of the surface gravity and the notion of stability for a WIH. Given these constraints, our first main result identifies a unique foliation by quasi-round MOTS. 
 \begin{theorem*}\textbf{\ref{t1}}
If $\mathcal{H}$ is a strictly stable, optically rigid WIH, then it admits a unique foliation by quasi-round MOTSs.
 \end{theorem*}
 In the second part of this paper we depart from NEHs globally and we concentrate on a strictly stable MOTS that `infinitesimally generates a NEH'. \\
 \indent In \cite{M1}, M. Mars shows that the structure of a sufficiently constrained NEH imposes the stability of any one of its cross-sections upon all others. This is how we make sense of the stability of $\mathcal{H}$ in Theorem \ref{t1}. In-fact, it is precisely the stability of a given cross-section, characterized by the invertibility of an elliptic operator on $\mathbb{S}^2$ (see Section 2.2), that allows us to solve and obtain a unique foliation by quasi-round MOTSs. \\
 \indent In work by L. Andersson-M. Mars-W. Simon in \cite{andersson2008}, the authors observe, whenever a hypersurface contains a MOTS that is strictly stabile relative to transverse variations within the hypersurface, then the existence of a MOTS within a perturbed hypersurfaces follows by way of an Implicit Function Theorem. Analogously in our setting, as a result of the strict stability of a given quasi-round MOTS, we obtain linear invertibility in the future null direction under the quasi-round constraint, and linear invertibility in the past null direction under the MOTS constraint. Since our unperturbed quasi-round MOTS is co-dimension two, we are able to combine both of these into an Implicit Function Theorem to guarantee a quasi-round MOTS for small metric pertrubations.
 \begin{theorem*}\textbf{\ref{t3}}
Consider a neighborhood $\mathcal{V}:=(-a,a)\times(-b,b)\times\mathbb{S}^2$, $a,b>0$ and a smooth Lorentzian metric $g_0\in \text{Sym}(T^\star\mathcal{V}\otimes T^\star\mathcal{V})$ satisfying the Dominant Energy Condition. Assume also that $\Sigma_0\cong \{0\}\times\{0\}\times\mathbb{S}^2$ is a strictly stable quasi-round MOTS in the geometry $(\mathcal{V},g_0)$ satisfying $\delta_{L^+}\langle\vec{H},\vec{H}\rangle=0$. Then, given a smooth path of Lorentzian metrics $\lambda\to g_\lambda\in \text{Sym}(T^\star\mathcal{V}\otimes T^\star\mathcal{V})$, $\lambda\in [0,c)\subset\mathbb{R}$, each $g_\lambda$ satisfying the Dominant Energy Condition, there exists an $0<\epsilon<c$ and a unique smooth path of embeddings, $\phi_\lambda:\mathbb{S}^2\hookrightarrow\mathcal{V}$, $0\leq \lambda\leq \epsilon$, such that each $\Sigma_\lambda:=\phi_\lambda(\mathbb{S}^2)$ is a quasi-round MOTS in the geometry $(\mathcal{V},g_\lambda)$.
 \end{theorem*}
Finally, after imposing `reasonable' decay similar to that of S. Alexakis \cite{A}, we show:
\begin{theorem*}\textbf{\ref{t4}}
Consider the Schwarzschild spacetime metric in ingoing Eddington-Finkelstein coordinates $(v,r,\vartheta,\varphi)$:
$$g_0 = -(1-\frac{2M}{r})dv\otimes dv+(dr\otimes dv+dv\otimes dr)+r^2\big(d\vartheta\otimes d\vartheta+(\sin\vartheta)^2d\varphi\otimes d\varphi\big),$$
on the neighborhood $\mathcal{U}:=(-\epsilon_0+v_0,\epsilon_0+v_0)\times(r_0,\infty)\times\mathbb{S}^2$, $\epsilon_0,r_0>0$. Consider also a smooth path of metrics, $\lambda\to g_\lambda\in \text{Sym}(T^\star\mathcal{U}\otimes T^\star\mathcal{U})$, $0\leq \lambda \leq c$, satisfying the Dominant Energy Condition. For $\{\Sigma_\lambda\}_{0\leq \lambda\leq \epsilon}$ the corresponding family of smooth quasi-round MOTS of Theorem \ref{t3}, whereby $\Sigma_0=\Omega_S\cap\mathcal{H}_S\cong\{v_0\}\times\{2M\}\times\mathbb{S}^2$ is a standard quasi-round MOTS of Schwarzschild, we assume the existence of an $\epsilon_1\leq \epsilon$ such that the past directed Null Cones, $\Omega_\lambda\supset\Sigma_\lambda$, $0\leq \lambda\leq \epsilon_1$, exist satisfying the curvature decay conditions (\ref{e18})-(\ref{e21}). Then, there exists an $0<\epsilon_2\leq \epsilon_1$ such that the Null Penrose Inequality
$$\sqrt{\frac{|\Sigma_\lambda|}{16\pi}}\leq m_{TB}(\lambda)$$
holds for $0\leq \lambda\leq \epsilon_2$.
\end{theorem*}
\section{Null Geometry and the Structure Equations}
As mentioned previously, in this paper we will be working with two different types of null hypersurface, namely \textit{Non Expanding Horizons} (or NEHs) and \textit{Null Cones}. It will be useful to introduce the notion of a general null hypersurface before specializing to our two examples.\\
\indent We take $\mathcal{N}$ to be a smooth, orientable, and connected hypersurface embedded in $(\mathcal{M},g)$. We say $\mathcal{N}$ is a \textit{null hypersurface} whenever the induced metric $\gamma:=g|_{\mathcal{N}}$ is degenerate. Equivalently, orientability of $\mathcal{N}$ ensures a smooth, non-vanishing vector field $\ubar L\in\Gamma(T\mathcal{N})$, such that $X\in\Gamma(T\mathcal{N})$ if and only if $\langle \ubar L,X\rangle = 0$. For the time being, we refrain from specifying whether $\ubar L$ is future, or past pointing. Since the ambient metric $g$ is non-degenerate, it follows that $\text{span}(\ubar L_p) = T_p^\perp\mathcal{N}\subset T_p\mathcal{N}$, for any $p\in\mathcal{N}$. Since $\mathcal{N}$ is a hypersurface, any $p\in\mathcal{N}$ admits a neighborhood $U_p\subset \mathcal{M}$, with a smooth function $v$ on $U_p$ (denoted $v\in\mathcal{F}(U_p)$) such that $V_p:=\mathcal{N}\cap U_p = \{v=0\}$, and the gradient $\text{grad}(v)\in\Gamma(T^\perp V_p)$ is nowhere vanishing. It follows that $\ubar L = f_1\text{grad}(v)$ for some smooth $f_1\neq 0$ on $V_p$, which implies $\langle \text{grad}(v), \text{grad}(v)\rangle|_{V_p} \equiv 0$. From the famous identity, $D_{\text{grad}(v)} \text{grad}(v) = \frac12\text{grad}\langle \text{grad}(v), \text{grad}(v)\rangle$, it follows that $\langle X,D_{\text{grad}(v)}\text{grad}(v)\rangle = 0$ for any $X\in\Gamma(T V_p)$. Therefore, $D_{\text{grad}(v)}\text{grad}(v) = f_2\text{grad}(v)$ on $V_p$, giving in turn $D_{\ubar L}\ubar L = \kappa\ubar L$, for smooth functions $f_2,\kappa\in\mathcal{F}(V_p)$. This allows us to conclude that integral curves of $\ubar L$, under a suitable re-parametrization, are null geodesics of $\mathcal{M}$ ruling the hypersurface $\mathcal{N}$. In other-words, $\mathcal{N}$ is a `congruence of null geodesics'.  \\
\indent We will also be imposing the existence of an embedded 2-sphere, $\iota:\mathbb{S}^2\hookrightarrow\mathcal{N}$, $\iota(\mathbb{S}^2)=\Sigma_0$, such that $\gamma|_{\Sigma_0}$ is a Riemannian or spacelike metric. We assume that any integral curve of $\ubar L$ intersects $\Sigma_0$ precisely once. This gives rise to a natural submersion $\pi:\mathcal{N}\to\Sigma_0$ sending $p\in\mathcal{N}$ to the intersection with $\Sigma_0$ of the integral curve $\beta_p^{\ubar L}$ of $\ubar L$, for which $\beta_p^{\ubar L}(0) = p$. Given $\ubar L$ and a constant $s_0$, we construct a smooth function $s\in\mathcal{F}(\mathcal{N})$ from the assignment that $\ubar L(s)=1$, and $s|_{\Sigma_0}= s_0$. For $q\in\Sigma_0$, if $(s_-(q),s_+(q))$ represents the range of $s$ along $\beta_q^{\ubar L}$, denoting ${S_-}:=\sup_{\Sigma}s_-$, and ${S_+}:=\inf_{\Sigma}s_+$, we notice that the interval $(S_-,S_+)$ is non-empty. Given that $\ubar L(s)=1$, the Implicit Function Theorem implies for $t\in(S_-,S_+)$, a spacelike embedding $\Sigma_t:= \{p\in\mathcal{N}|s(p)=t\}$ with diffeomorphism $\pi|_{\Sigma_t}:\Sigma_t\to\Sigma_0$. We refer to such spherical embeddings as \textit{cross-sections} of $\mathcal{N}$. As in Section 1.1, with a slight abuse of notation, we will also drop subscripts whenever referring to the induced metric on an arbitrary cross-section $\Sigma\hookrightarrow\mathcal{N}$, hence $\Sigma \cong (\mathbb{S}^2,\gamma)$. We note for $s<S_-$ or $s>S_+$, in the case that $\Sigma_s$ is non-empty, such surfaces remain smooth but may no longer be connected. Nevertheless, the collection $\{\Sigma_s\}_s$ gives a foliation of $\mathcal{N}$. We also highlight that the foliation $\{\Sigma_s\}_s\subset \mathcal{N}$ is dependent upon our choice of $\ubar L$. Any other viable candidate would have to be a re-scaling of $\ubar L$, say to $\ubar L_a:=a\ubar L$, for some $a\neq 0$, whereby $D_{\ubar L_a}\ubar L_a = a(\ubar L\log |a|+\kappa)\ubar L_a$. The corresponding foliation, $\{\Sigma^a_s\}\subset \mathcal{N}$, is a \textit{geodesic foliation} whenever $\ubar L\log |a| = -\kappa$, equivalently, $a(s,q) = a_0(q)e^{-\int_{s_0}^s\kappa(u,q)du}$, for $q\in \Sigma_{s_0}$.\\\\
\indent  A significant amount of our efforts will go towards analysing data on individual cross-sections, $\Sigma\subset\mathcal{N}$. The local ambient geometry extrinsic to $\mathcal{N}$ is observed through the adapted null vector field $L\in \Gamma(T^\perp \Sigma)$. For convenience, we remind the reader that we choose $L$ by assigning to every $q\in\Sigma$, the unique null vector satisfying $\langle \ubar L_q,L_q\rangle = 2$, and $\langle L_q,v\rangle = 0$, where $v\in T_q\Sigma$. Also, we recall the symmetric 2-tensors $\ubar\chi,\chi$ and the connection 1-from $\zeta$ with respect to $\{\ubar L,L\}\subset\Gamma(T^\perp\Sigma)$. For a foliation $\{\Sigma_s\}_s$ associated with $\ubar L$, we therefore similarly construct an associated null vector field $L_s$.\\
\indent Given a cross-section $\Sigma\subset\mathcal{N}$, and $v\in T_q(\Sigma)$ we may extend $v$ along the generator $\beta_q^{\ubar L}$ according to the ODE:
\begin{align*}
\dot{V}(s) &= D_{V(s)}\ubar L\\
V(0)&=v.
\end{align*} 
\indent As a result, we have $\frac{d}{ds}{\langle V(s),\ubar L\rangle}=\frac12V(s)\langle\ubar L,\ubar L\rangle+\langle V(s),D_{\ubar L}\ubar L\rangle=\kappa\langle V(s),\ubar L\rangle$, and since $v\in T_p\mathcal{N}\iff \langle \ubar L|_p,v\rangle = 0$, it follows that $\langle V(0),\ubar L_p\rangle=0$. By ODE uniqueness, we conclude $\langle V(s),\ubar L\rangle=0$ for all $s$. As a result, any section $W\in\Gamma(T\Sigma)$ may be extended througout $\mathcal{N}$ satisfying $[\ubar L,W]=0$. Along each generator, $0=[\ubar L,W]s = \ubar L(Ws)=\frac{d}{ds}(Ws)$, so that $Ws|_\Sigma=0$ ensures $Ws = 0$ throughout $\mathcal{N}$. We conclude along the foliation $\{\Sigma_s\}_s\subset\mathcal{N}$, that $W|_{\Sigma_s}\in\Gamma(T\Sigma_s)$, and denote by $E(\Sigma)\subset\Gamma(T\mathcal{N})$ the set of such extensions off of $\Sigma$ along $\ubar L$.  We also note that linear independence is preserved along generators by standard ODE uniqueness theorems, allowing us to extend a local basis $\{X_1,X_2\}\subset\Gamma(TU)$ ($U\subset \Sigma$), throughout $\pi^{-1}(U)$. Having established a background foliation $\{\Sigma_s\}_s$, the fact that $\mathcal{N}$ is generated by null pregeodesics along $\ubar L$ then uniquely characterizes any spacelike cross-section $\Sigma\subset\mathcal{N}$ as a graph over $\Sigma_0$ with graph function $s\circ (\pi|_{\Sigma})^{-1} = \omega\in\mathcal{F}(\Sigma_0)$. By Lie-dragging $\omega$ along $\ubar L$ to all of $\mathcal{N}$, we have for any $V\in E(\Sigma_0)$, that $(V+V\omega\ubar L)(s-\omega) = 0$, from which we conclude that $ V_\omega:=V+V\omega \ubar L$ satisfies $V_\omega|_\Sigma\in\Gamma(T\Sigma)$. We will now slightly abuse notation by ignoring obvious restrictions to $\Sigma$. Taking $L_\omega:=L_s+|\nabla\omega|^2\ubar L-2\nabla\omega$, it follows that $\langle L_\omega,V_\omega\rangle=\langle L_s,V\omega\ubar L\rangle-2\langle\nabla\omega,V_\omega\rangle = 2V\omega-2 V_\omega\omega=0$, therefore $L_\omega\in\Gamma(T^\perp\Sigma)$. As a result, $0=\langle L_\omega,\nabla\omega\rangle = \langle L_s,\nabla\omega\rangle-2|\nabla\omega|^2$, and we deduce therefore that $\langle L_\omega,L_\omega\rangle = 4|\nabla\omega|^2+2\langle L_s,|\nabla\omega|^2\ubar L\rangle-4\langle L_s,\nabla\omega\rangle=0$. So in-fact, $L_\omega: = L_s+|\nabla\omega|^2\ubar L-2\nabla\omega$ identifies again the previously constructed null partner to $\ubar L$ in the normal bundle of $\Sigma$. This time with respect to the background foliation $\{\Sigma_s\}_s$. It will also be helpful to define the pointwise projection of $\nabla\omega$ to the background foliation, specifically for any $p\in\Sigma\cap\Sigma_{s(p)}$, we take $\nabla_s\omega|_p:=(\nabla\omega-\frac12\langle L_s,\nabla\omega\rangle\ubar L)|_p\in T_p\Sigma_{s(p)}$. We see $\langle V|_p,\nabla_s\omega|_p\rangle = (V\omega)(p)$, so as suggested by the notation, $\nabla_s\omega$ agrees pointwise with the gradient on $\Sigma_s$ of the function $\omega|_{\Sigma_s}$.

\begin{remark}
We bring to the attention of the reader that $\ubar\chi$ is in-fact independent of the cross-section $\Sigma\subset\mathcal{N}$. Specifically, ignoring obvious restrictions for brevity, we observe $\ubar\chi(V_\omega,W_\omega) = \langle D_{V+V\omega\ubar L}\ubar L,W+W\omega\ubar L\rangle = \langle D_V\ubar L,W\rangle+W\omega\langle D_{V+V\omega\ubar L}\ubar L,\ubar L\rangle+V\omega\langle D_{\ubar L}\ubar L,W\rangle=\langle D_V\ubar L,W\rangle = \ubar\chi_s(V,W)$. Since any cross-section can be realized as the leaf of an appropriately chosen geodesic foliation, it follows that the non-trivial components of the second fundamental form of $\mathcal{N}$ at $p\in\mathcal{N}$ is fully determined by $\ubar\chi|_p$ for an arbitrary cross-section $\Sigma\subset\mathcal{N}$ containing $p$. This is also clearly the case for the induced metric $\gamma$ on $\mathcal{N}$. We conclude therefore, at any point $p\in\mathcal{N}$, the trace $\tr\ubar\chi(p)$ defines a smooth function on $\mathcal{N}$ dependent only on $\ubar L$.
\end{remark}

Considering this remark, similarly as for the metric $\gamma$ on $\mathcal{N}$, we will interchangeably denote by $\ubar \chi$ the second fundamental form of $\mathcal{N}$, and its restriction to a cross-section $\Sigma\subset\mathcal{N}$. \\

Suppose $\mathcal{N}$ is endowed with a foliation $\{\Sigma_s\}_s\subset\mathcal{N}$ along $\ubar L$. If the associated data is given by $(\gamma_s,\ubar\chi_s,\chi_s,\zeta_s)$ on $\Sigma_s$, then, for a cross-section $\Sigma:=\{s=\omega\}$, $\omega\in\mathcal{F}(\Sigma_0)$ we have:

\begin{lemma}\label{l1}
Suppose $\mathcal{N}$ is endowed with a foliation $\{\Sigma_s\}_s\subset\mathcal{N}$ along $\ubar L$. At a point $p\in\Sigma\cap\Sigma_{s(p)}$ we have by a slight abuse of notation: 
\begin{align*}
\zeta(V_\omega) &= \zeta_s(V)-\ubar\chi_s(V,\nabla_s\omega)+\kappa\langle V,\nabla_s\omega\rangle,\\
\tr\chi&=\tr \chi_s-4\zeta_s(\nabla_s\omega)-2\Delta\omega+\tr \ubar\chi_s|\nabla\omega|^2-2\kappa|\nabla\omega|^2.
\end{align*}
\end{lemma}
\begin{proof}
See, for example \cite{R,S}.
\end{proof}\newpage
\begin{proposition}[Structure Equations]\label{p3} Along a foliation $\{\Sigma_s\}_s\subset\mathcal{N}$, with data $(\gamma_s,\ubar\chi_s,\chi_s,\zeta_s)$ associated to $\ubar L$:
\begin{align}
\ubar L\mathcal{K}_s &= -\tr\ubar\chi_s\mathcal{K}_s-\frac12\Delta_s \tr\ubar\chi_s+\nabla_s\cdot(\nabla_s\cdot\hat{\ubar\chi}_s)\\
\pounds_{\ubar L}\gamma_s &= 2\ubar\chi_s\label{e4}\\
\pounds_{\ubar L}\ubar\chi_s &= -\ubar\alpha_s + \frac12|\hat{\ubar\chi}_s|^2\gamma_s + \tr\ubar\chi_s\hat{\ubar\chi}_s+\frac14(\tr\ubar\chi_s)^2\gamma_s + \kappa\ubar\chi_s\\
\ubar L\tr\ubar\chi_s &= -\frac12(\tr\ubar\chi_s)^2 - |\hat{\ubar\chi}_s|^2 - G(\ubar L,\ubar L) +\kappa \tr\ubar\chi_s\label{e6}\\
\pounds_{\ubar L}\chi_s&=\Big(\mathcal{K}_s+\hat{\ubar\chi}_s\cdot\hat{\chi}_s+\frac12G(\ubar L,L_s)\Big)\gamma_s+\frac12\tr\ubar\chi_s\hat{\chi}_s+\frac12\tr\chi_s\hat{\ubar\chi}_s\\
&\qquad-\hat{G}_s-2\text{Sym}(\nabla_s\zeta_s)-2\zeta_s\otimes\zeta_s-\kappa\chi_s\notag\\
\ubar L\tr\chi_s &= G(\ubar L,L)+2\mathcal{K}_s-2{\nabla}_s\cdot\zeta_s-2|\zeta_s|^2-\langle\vec{H}_s,\vec{H}_s\rangle-\kappa\tr\chi_s\label{e8}\\
\pounds_{\ubar L}\zeta_s &= G_{\ubar L} - \nabla_s\cdot\hat{\ubar\chi}_s - \tr\ubar\chi_s\zeta_s+\frac12{d_s}\tr\ubar\chi_s+{d_s}\kappa\label{e9}
\end{align}
where $\nabla_s$ denotes the induced covariant derivative on $\Sigma_s$, $\mathcal{K}_s$ denotes the Gauss curvature of $\Sigma_s$, $\ubar\alpha_s$ is the symmetric 2-tensor given by $\ubar\alpha_s(V,W) = \langle R_{\ubar L V}\ubar L,W\rangle$, $\text{Sym}(T)$ represents the symmetric part of a 2-tensor $T$, $G_{\ubar L} = G(\ubar L,\cdot)|_{\Sigma_s}$, and $\hat{G}_s = G|_{\Sigma_s}-\frac12(\tr_{\gamma}G)\gamma$.
\end{proposition}
\begin{proof}
See, for example \cite{R,G}.
\end{proof}
Armed with the Structure equations of $\mathcal{N}$, we are ready to distinguish between Non Expanding Horizons and Null Cones. For either structure we will henceforth assume the existence of a rescaling $\ubar L_a=a\ubar L$, $a>0$, such that $S_+(a)=\infty$, and $\ubar L_a$ is geodesic (i.e. $\ubar L\log a = -\kappa$, $D_{\ubar L_a}\ubar L_a=0$).
\begin{definition}\label{d8}
We say $\mathcal{N}$ is a Null Cone, denoted $\Omega$, if $\tr\ubar\chi>0$ throughout $\Omega$.
\end{definition}

From equation (\ref{e4}), we also observe $\frac{d}{ds}(dA_s) = \tr\ubar\chi_s dA_s$ for the area form $dA_s$ on $\Sigma_s$, so along $\ubar L$ we have pointwise area expansion. For any re-scaling $\ubar L_a = a\ubar L$, $a>0$, we consequently observe $\tr\ubar\chi_a = a\tr\ubar\chi>0$, and the same pointwise area expansion follows for an associated foliation $\{\Sigma_{s_a}^a\}$ hence the name \textit{Null Cone}. From now on, we will always assume that $\ubar L$ is a past-pointing null pregeodesic generator of $\Omega$, so that $\Omega$ is a past-directed Null Cone.

\begin{figure}[h]
\centering
\includegraphics[width=0.75\textwidth]{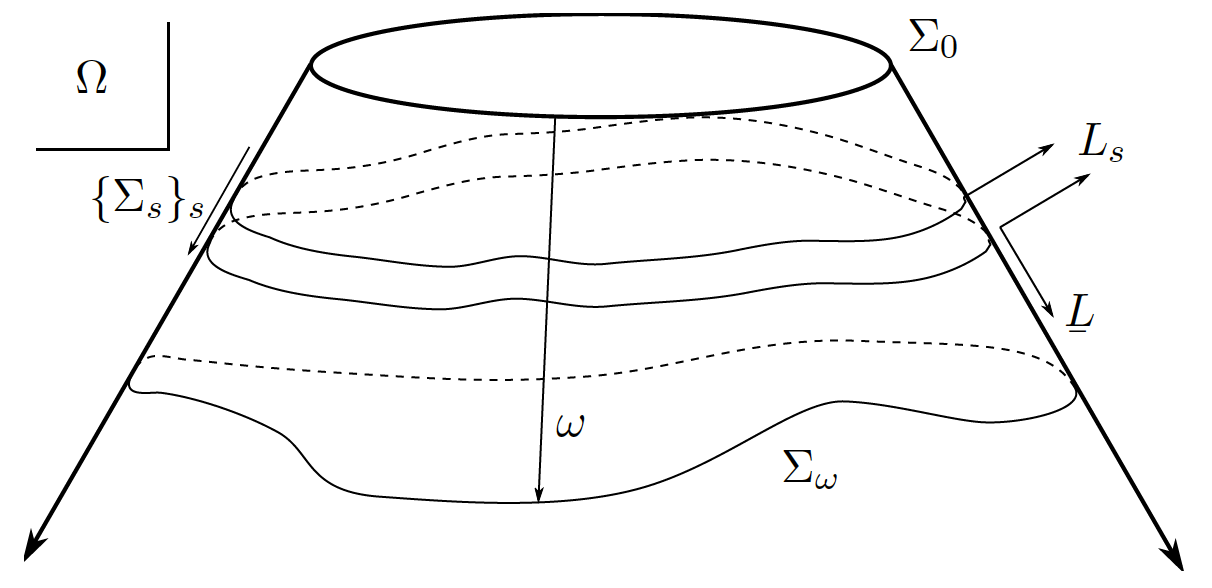}
\end{figure}
\newpage
\begin{definition}\label{d9}
We say $\mathcal{N}$ is a Non Expanding Horizon (NEH), denoted $\mathcal{H}$, if $\tr\ubar\chi = 0$ throughout $\mathcal{H}$.
\end{definition}
Onwards, when considering a NEH $\mathcal{H}$, we will always assume that $\ubar L$ is a future-pointing pregeodesic generator. We observe in a spacetime, this identifies a hypersurface for which any cross-section exhibits pointwise area conservation along orthogonal light ray propagation. In other-words, light rays are marginally trapped.
\begin{lemma}\
\begin{enumerate}
\item	$\mathcal{N}$ is a Null Cone if and only if there exists a cross-section $\Sigma\subset\mathcal{N}$ expanding along $\ubar L$, i.e. $\tr\ubar\chi|_\Sigma>0$.
\item  $\mathcal{N}$ is a NEH if and only if there exists a cross-section $\Sigma\subset\mathcal{N}$ that is marginally trapped along $\ubar L$, i.e. $\tr\ubar\chi|_{\Sigma}=0$.
\end{enumerate}
\end{lemma}
\begin{proof}
 Without loss of generality, we take $\ubar L$ to be a geodesic generator with a given geodesic $\beta:=\beta_p^{\ubar L}$, giving $\kappa\equiv 0$. Provided $|\tr\ubar\chi\circ\beta(s_1)|>0$ for some $s_1$, we have from (\ref{e6}):
$$\frac{d}{ds}\frac{1}{\tr\ubar\chi}\circ \beta(s) = \frac12+\frac{|\hat{\chi}|^2+G(\ubar L,\ubar L)}{\tr\ubar\chi^2}\circ\beta(s)\geq \frac12\implies \frac{1}{\tr\ubar\chi}\circ\beta(s)-\frac{1}{\tr\ubar\chi}\circ\beta(s_1)\geq \frac12(s-s_1)$$
as long as $|\tr\ubar\chi\circ\beta(s)|>0$, for all $s\geq s_1$. It follows that our assumption $S_+ = \infty$ rules out the possibility that $\tr\ubar\chi\circ\beta(s_1)<0$, since otherwise we observe $\lim_{s\to s_2^-}\tr\ubar\chi\circ\beta(s) = -\infty$ for some $s_1<s_2<\infty$, in contradiction to the smoothness of $\mathcal{N}$. Therefore, $\tr\ubar\chi\circ\beta\geq 0$ throughout $\mathcal{N}$. We now specialize to the two cases above:
\begin{enumerate}
\item If we parametrize so that $\beta(s_0)\in \Sigma$, (\ref{e6}) immediately implies that $\tr\ubar\chi\circ\beta(s)>0$ for $s<s_0$. Moreover, $\frac{1}{\tr\ubar\chi}\circ\beta(s)\geq \frac12(s-s_0)+\frac{1}{\tr\ubar\chi}\circ\beta(s_0)>0$ giving $\tr\ubar\chi\circ\beta(s)>0$ for $s\geq s_0$. 
\item If $\tr\ubar\chi\circ\beta(s_1)>0$ for any $s_1$, then we've already argued that $\tr\ubar\chi\circ\beta>0$ contradicting the hypothesis. It therefore follows that $\tr\ubar\chi\circ\beta(s)=0$ for all $s$.
\end{enumerate}
\end{proof}
For a NEH $\mathcal{H}$, in contrast to a Null Cone $\Omega$, it will also be helpful to denote the future pointing pregeodesic generator by $l$ instead of $\ubar L$. For an arbitrary cross-section $\Sigma\subset\mathcal{H}$ we will also denote the transverse null vector field by $k$ instead of $L$ (recall $\langle l,k\rangle = 2$), the associated data will be denoted:
$$\ubar\chi\to\chi_l,\,\,\tr\ubar\chi\to\theta_l,\,\,\zeta\to -t,\,\,\chi\to\chi_k,\,\,\tr\chi\to\theta_k.$$
From our Lemma above, it follows that $\theta_l\equiv 0$ on $\mathcal{H}$.\\
\indent From equation (\ref{e6}), also known as the famous \textit{optical Raychaudhuri equation}, we notice that since $\theta_l\equiv0$ we have $|\hat{\chi}_l|^2+G(l,l)\equiv 0$, giving $\chi_l = 0$, and $G(l,\cdot)=f_3\langle l,\cdot\rangle$ for some $f_3\in\mathcal{F}(\mathcal{H})$. We conclude that $\langle D_XY,l\rangle = -\chi^l(X,Y) = 0$ for any $X,Y\in\Gamma(T\mathcal{H})$. It follows therefore that $D_XY\in\Gamma(T\mathcal{H})$, so that the ambient connection $D$ restricts to a connection on $\mathcal{H}$. We also observe, as a consequence of equation (\ref{e4}), any two cross-sections $\Sigma_1,\Sigma_2\subset\mathcal{H}$ are isometric.
\begin{remark}\label{r3}
Given a background foliation $\{\Sigma_s\}_s\subset \mathcal{H}$, and a basis extension $\{X_i\}\subset E(U)$ ($U\subset\Sigma_0$), we can define $\vec{G}_l:=G(l,X_i)\gamma^{ij}X_j$ and therefore the vector field $\frac14l-\epsilon\vec{G}_l-\epsilon^2|\vec{G}_l|^2k_s$ is future pointing and null for any $\epsilon,s$. From the fact that $G(l,l)=0$, we have for any $\epsilon>0$ according to the DEC that
$$G(l,\frac14l-\epsilon\vec{G}_{l}-\epsilon^2|\vec{G}_{l}|^2k)=-\epsilon|\vec{G}_{l}|^2(1+\epsilon G(l,k))\geq0.$$
This is impossible unless also $\vec{G}_l\equiv0$, giving $G(l,X) = 0$ for any $X\in\Gamma(T\mathcal{H})$.
\end{remark}

\begin{figure}[h]
\centering
\includegraphics[width=0.4\textwidth]{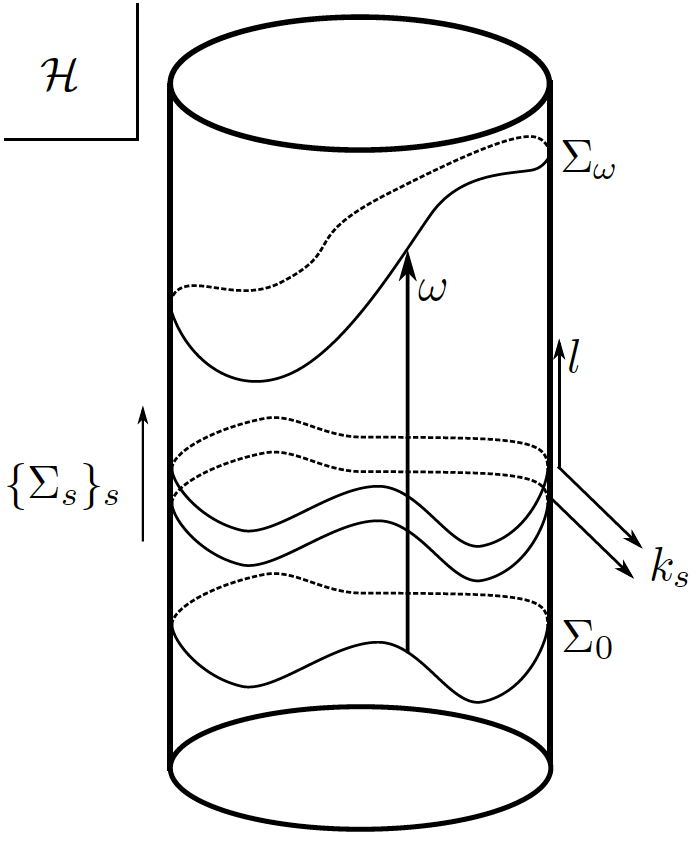}
\end{figure}

\subsection{Weakly Isolated Horizons}
\indent NEHs have been extensively studied in the literature \cite{PhysRevLett.85.3564, PhysRevD.64.044016, 0264-9381-19-6-311, PhysRevD.62.104025, G, Hai1973, Ja} arising as a quasi-local model for a black hole event horizon in relativity theory. Since the precise location for an event horizon requires a full understanding of the future evolution of the spacetime, a quasi-local model serves as a convenient approximation. Arguably the most prolific and most restrictive example is a \textit{Killing Horizon}. A Killing Horizon is a null hypersurface with null pregedesic generator given by $l=\xi|_\mathcal{H}$, for $\xi$ a Killing field of $\mathcal{M}$. Satisfying the Killing equation:  
$$\langle D_X\xi,Y\rangle+\langle D_Y\xi,X\rangle = 0,$$
for any $X,Y\in\Gamma(T\mathcal{M})$, we immediately observe Definition \ref{d9} and its consequences on $\chi_l, G(l,\cdot)$ when we restrict to $\mathcal{H}$. Killing Horizons arise from the study of stationary space-times in general relativity. While an invaluable idealization, stationary spacetimes do not allow for gravitational radiation or collapse. It is strongly expected that isolated gravitating systems will rapidly approach a stationary spacetime as the equilibrium state, nonetheless the space-time will not be stationary. An equilibrium black hole horizon in the NEH family constructed independently of the ambient geometry was given in \cite{PhysRevLett.85.3564}, called an \textit{Isolated Horizon} (IH) (see also \cite{doi:10.1139/p05-063}). Specifically, an IH is a NEH admitting a pregeodesic generator $l$ such that the 2-tensor $[\pounds_l,D]$ satisfies: 
$$([\pounds_l,D]X)(Y):=[l,D_YX]-D_{[l,Y]}X-D_Y[l,X]=0$$
for any $X,Y\in\Gamma(T\mathcal{H})$. Similarly to a Killing Horizon, an IH observes a preferred pregeodesic generator $l$ (up to constant re-scaling) along which the induced metric and connection are `time-independent'. In contrast to a Killing Horizon, the generator $l$ for an IH holds no a-priori constraints from a local extension in a neighborhood of $\mathcal{H}$. Spacetimes admitting an isolated horizon need not admit any Killing fields, see \cite{Lewandowski_2000}. The famous Robinson-Trautman spacetime supports a global isolated horizon, $\mathcal{H}$, with gravitational radiation in every neighborhood of $\mathcal{H}$, see \cite{chrusciel1992RobTra}. In \cite{PhysRevLett.85.3564}, it was also observed that IHs give rise to a well defined action principle and Hamiltonian formalism. Thus allowing for an extension of the famous laws of black hole thermodynamics from Killing Horizons to IHs. In order to use as few as possible assumptions on our NEH $\mathcal{H}$, we can further broaden the class of viable choices beyond that of an IH. The class of NEH that weakens the constraints of an IH just beyond those we'll need to prove existence of a quasi-round foliation, is interestingly also the weakest construction that allows one to extend the zeroth law of black hole thermodynamics, namely that of constant surface gravity. 
\begin{definition}\label{d10} (see \cite{0264-9381-19-6-311})\
 We say a NEH $\mathcal{H}$ is a Weakly Isolated Horizon (WIH) if we can find a pregeodesic generator, $l$, such that the tensor $[\pounds_l,D]l$ is trivial.
\end{definition}
\begin{lemma}\label{l2}
A pregeodesic generator, $l$, of a NEH, $\mathcal{H}$, defines a WIH structure if and only if its associated surface gravity $\kappa_l$ is constant on $\mathcal{H}$.
\end{lemma}
\begin{proof}
We show $[\pounds_l,D]l=d\kappa_l\otimes l$. It suffices to verify this equality for a convenient frame at any $q\in\mathcal{H}$. We choose this frame using a local extension basis $\{X_i\}$, and $l$ itself. Firstly we see
$$([\pounds_l, D]l)(X_i) = [l,D_{X_i}l]=[l,t(X_i)l]=X_i\kappa_ll,$$
having used $D_{X_i}l = (\chi_l)_{ij}\gamma^{jk}X_k+t(X_i)l = t(X_i)l$ in the second equality, and (\ref{e9}) in the third. Similarly, 
$$([\pounds_l,D]l)(l)=[l,D_ll] = [l,\kappa_ll]=l\kappa_ll.$$
\end{proof}
As a result of Lemma \ref{l1}, a NEH $\mathcal{H}$ can easily be given a WIH structure for any choice of $\kappa_l$. By taking a geodesic tangent vector field, $l'$, with some associated foliation $\{\Sigma'_s\}\subset\mathcal{H}$, we obtain our desired surface gravity $\kappa_l$ from $l:=fl'$ whereby $f = \kappa_ls+\phi$ for some $\phi\in\mathcal{F}(\Sigma_0)$. This freedom is not only too broad, the non-uniqueness is not particularly useful for physical applications (see \cite{0264-9381-19-6-311}, Section IV). So as in \cite{0264-9381-19-6-311}, we will need a more refined WIH structure on $\mathcal{H}$. 
\begin{definition}\label{d11}
We say a WIH structure on a NEH, $\mathcal{H}$, is optically rigid if:
\begin{enumerate}
\item The pregeodesic generator $l$ induces a foliation such that $lG(l,k_s)=0$
\item There exists a cross-section, $\Sigma\subset\mathcal{H}$, such that $\theta_k>0,\,l\theta_k=0$.
\end{enumerate} 
\end{definition}
Before un-packing Definition \ref{d11}, we will need the following consequence of Lemma \ref{l1}:
\begin{lemma}(\cite{M1}, Lemma 3)\label{l3}\\
Consider a NEH with background foliation $\{\Sigma_s\}_s\subset\mathcal{H}$, and the cross-section $\Sigma:=\{s=\omega\}\subset \mathcal{H}$, $\omega\in\mathcal{F}(\Sigma_0)$, with adapted null normal $k_\omega\in\Gamma(T^\perp\Sigma)$. Then the map $\pi|_\Sigma:\Sigma\to\Sigma_{0}$ is an isometry, and:
\begin{align}
t_\omega(V_\omega) &= t_s(V)-\kappa_l V\omega\\
(\theta_k)_\omega&=\theta_k - 2\Delta\omega-2\kappa_l|\nabla\omega|^2+4t_s(\nabla_s\omega)\label{e11}.
\end{align} 
\end{lemma}
We start by noticing the first condition in Definition \ref{d11} is independent of $l$. In-fact, from an arbitrary cross-section $\Sigma:=\{s=\omega\}\subset\mathcal{H}$, $\omega\in\mathcal{F}(\Sigma_0)$, we may initiate a flow along $l$, $\Sigma_t:=\{s=\omega+t\}$, $t\geq0$. Each $\Sigma_t$ has adapted null normal $k_{\omega+t}=k_s+|\nabla(\omega+t)|^2l-2\nabla(\omega+t) = k_s+|\nabla\omega|^2l-2\nabla\omega$. As a result of Remark \ref{r3}, $G(l,k_{\omega+t}) = G(l,k_s)$, thus $lG(l,k_{\omega+t}) = lG(l,k_s)=0$. Moreover, for a re-scaling $l\to fl$ with associated flow $\Sigma_u:=\{s=\omega_u\}$, the adapted null normal is of the form $k_u:=\frac{1}{f}(k_s+|\nabla\omega_u|^2l-2\nabla\omega_u)$ giving $G(fl,k_u) = G(l,k_s)$. Therefore, again $flG(fl,k_u) = flG(l,k_s)=0$.\\
\indent 
Combining the fact that $\kappa_l$ is constant and that $\pi$ restricts to an isometry on arbitrary cross-sections, we have from (\ref{e8},\ref{e9}) 
$$l^2\theta_k = lG(l,k)-\kappa_ll\theta_k=-\kappa_ll\theta_k.$$
So for any foliation $\{\Sigma_s\}_s$ along $l$, we observe $l\theta_{k_s} = (l\theta_{k_0})e^{-\kappa_ls}$. From the second condition of Definition \ref{d11} we have at least one foliation along $l$ for which $\theta_k$ remains constant. If we assume this particular foliation as a background foliation, then we may again initiate a flow along $l$ off of an arbitrary $\{s=\omega\}\subset\mathcal{H}$. From Lemma \ref{l3} we notice $(\theta_k)_{\omega+t} = (\theta_k)_\omega$. So all foliations along $l$ conserve $\theta_k$. Unlike the first condition, we will now find that re-scalings $l\to fl$ will modify $l\theta_k$ by a second order elliptic equation in the function $f$. This is to be expected from (\ref{e11}) and leads us to the \textit{Stability Operator} of a MOTS. Regarding the $\theta_k>0$ assumption in Definition \ref{d11}, we make this assumption in order to simplify the notion of stability for a WIH structure. Since we will also be dealing with cross-sections of $\mathcal{H}$ that can be connected, via a Null Cone $\Omega$ to past null infinity, Definition \ref{d8} compels that $\theta_k>0$ for some $\Sigma\subset\mathcal{H}$. Moreover, once we specify stability on an optically rigid WIH structure we will see that the associated pregeodesic $l$ is now unique (up to constant rescaling).

\subsection{The Stability Operator}
Assume we have a 2-sphere embedding $\iota:\mathbb{S}^2\hookrightarrow\mathcal{M}$ and a differentiable map $\Phi: \mathbb{S}^2\times(-\epsilon,\epsilon)\to \mathcal{M}$ such that $\Phi(\cdot, t)$ is an immersion and $\alpha_p(t):=\Phi(p,t)$, $p\in\mathbb{S}^2$, is a curve satisfying $\alpha_p(0) = \iota(p)$ with initial velocity $\alpha_p'(0) = \psi(p)(\frac12\ubar L_{\iota(p)}+\phi(p)L_{\iota(p)})$ for smooth functions $\psi,\phi\in\mathcal{F}(\mathbb{S}^2)$. If we now consider the family of 2-spheres, $\Sigma_t:=\Phi(\mathbb{S}^2,t)$, and a corresponding family $t\to f_t\in\mathcal{F}(\Sigma_t)$ differentiable in $t$, $f:=f_0$, we denote the variation:
$$\big(\delta_{\psi(p)(\frac12\ubar L_{\iota(p)}+\phi(p)L_{\iota(p)})}f\big)(p):=\frac{d}{dt}|_{t=0}(f_t\circ\Phi)(p,t).$$
Next we will need to introduce the notion of a stable MOTS. To this end, assume we have a 2-sphere $\Sigma\subset\mathcal{M}$ satisfying $\tr\chi =0$. Denoting $\nu:=\frac12\ubar L+\phi L$, it follows (see \cite{andersson2008}, Lemma 3.1) that the linearization satisfies $\delta_{\psi\nu}\tr\chi= \mathcal{L}_\nu(\psi)$:
\begin{definition}\label{d12} Given a MOTS $\Sigma$ in $\mathcal{M}$,
$$\mathcal{L}_\nu(\psi):=-\Delta\psi-2\nabla\cdot(\psi\zeta)+\Big(\mathcal{K}+\nabla\cdot\zeta-|\zeta|^2+G(\nu,L)+\phi|\hat{\chi}|^2\Big)\psi$$
is called the Stability Operator of $\Sigma$ along $\nu=\frac12\ubar L+\phi L$. It's adjoint with respect to the $L^2$ inner product on $\Sigma$ is also given by
$$\mathcal{L}_\nu^\star(\psi) = -\Delta\psi+2\zeta(\nabla\psi)+\Big(\mathcal{K}+\nabla\cdot\zeta-|\zeta|^2+G(\nu,L)+\phi|\hat{\chi}|^2\Big)\psi.$$
\end{definition}
Like any second order elliptic PDE on a compact Riemannian manifold, the stability operator $\mathcal{L}_\nu$ admits a principal real-valued eigenvalue $\mu$ with one-dimensional eigenspace of the form $\{c\varphi\}$ where $0<\varphi\in\mathcal{F}(\Sigma)$ is smooth, $c\in\mathbb{R}$. Moreover, both $\mathcal{L}_\nu$, $\mathcal{L}_\nu^\star$ share a common principal eigenvalue (see \cite{andersson2008}, Lemma 4.1). We say $\Sigma$ is \textit{stable (strictly stable)} if $\mu\geq(>)0$ and \textit{unstable} if $\mu<0$. We leave it to the reader to verify that any scale change $\ubar L\to a\ubar L$, $L\to \frac{1}{a}L$ for some smooth function $a>0$ results in the stability operator (along the same direction, i.e. $\phi\to a^2\phi$) changing to 
$$\mathcal{L}_\nu^a(\psi) = \frac{1}{a}\mathcal{L}_\nu(a\psi).$$
From this we also conclude that stability is invariant under scale variations since the new principle eigenfunction is given by $\varphi_a = \frac{1}{a}\varphi>0$. For further discussion regarding the notion of stability we refer the reader to \cite{andersson2008}. Throughout this paper, ${C}^{k,\alpha}(\Sigma)$ refers to the space of $k$-times differentiable functions on our 2-sphere $\Sigma$ with $k^{\text{th}}$ partial derivatives being H\"{o}lder continuous with exponent $0<\alpha\leq 1$. The (irrelevant for our purposes) H\"{o}lder seminorm $|\cdot|_\alpha$ depending on some fixed background metric (i.e. the standard round metric will do).
\begin{lemma}(\cite{andersson2008}, Lemma 4.2)\label{l4}\
Let $\mathcal{L}_\nu$ be the stability operator of a MOTS $\Sigma$. Let $\mu$ and $\varphi>0$ be the principal eigenvalue and eigenfunction of $\mathcal{L}_\nu$, respectively, and let $\psi\in C^{2,\alpha}(\Sigma)$ be a solution of $\mathcal{L}_\nu(\psi) = f$ for some function $0\leq f\in C^{0,\alpha}(\Sigma)$. Then the following holds,
\begin{enumerate}
\item If $\mu=0$, then $f\equiv 0$ and $\psi = C\varphi$ for some constant C.
\item If $\mu>0$ and $f\not\equiv0$, then $\psi>0$.
\item If $\mu>0$ and $f\equiv 0$, then $\psi\equiv0$.
\end{enumerate}
\end{lemma}
We observe that both the operators $\mathcal{L}_\nu,\mathcal{L}_\nu^\star: C^{2,\alpha}(\Sigma)\to C^{0,\alpha}(\Sigma)$ have smooth coefficients on a compact manifold and by standard results are therefore bounded. As a consequence of the Fredholm Alternative and the Bounded Inverse Theorem, Lemma \ref{l4} therefore ensures that whenever $\mu>0$ both operators have bounded linear inverses.\\\\
\indent As seen in \cite{M1}, we now notice that $\mathcal{L}_{\nu}=\mathcal{L}_{\frac12\ubar L}+\phi\Big(|\hat{\chi}|^2+G(L,L)\Big)$. We will always assume that the stability operator is defined for $\ubar L$ past-pointing. Therefore, for a cross-section of a NEH, $\Sigma\subset\mathcal{H}$, the stability operator is calculated for variations of $\theta_l$ along $\nu = \frac12k+\phi l$. Since $\chi_l\equiv 0$, $G(l,l)=0$, the stability of any cross-section is therefore fully described by the operator 
$$\mathcal{L}_{\frac12k}(\psi) = -\Delta\psi-2\nabla\cdot(\psi t)+\Big(\mathcal{K}+\nabla\cdot t-|t|^2+\frac12G(l,k)\Big)\psi.$$ 
We will only consider MOTSs satisfying $|\hat{\chi}|^2+G(L,L)=0$, so for the remainder of this work we will omit the subscripts on the stability operator for the purposes of identifying a direction of perturbation. \\\\
\indent As a result of the optically rigid structure of Definition \ref{d11} on a WIH, the stability operator simplifies using (\ref{e8}):
$$\mathcal{L}(\psi) = -\Delta\psi-2\nabla\cdot(\psi t)+\frac12\kappa_l\theta_k\psi.$$
We now see that the optically rigid WIH structure is enough to recover \cite{M1}, Proposition 3, and to use \cite{M1}, Proposition 4 to impose stability on all of $\mathcal{H}$.
\begin{proposition}\label{p4}
Consider an optically rigid WIH structure on a NEH, $\mathcal{H}$, and a cross-section, $\Sigma\subset\mathcal{H}$. Then,
\begin{enumerate}
\item If $\kappa_l>0$, $\Sigma$ is strictly stable,
\item If $\kappa_l=0$, $\Sigma$ is marginally stable,
\item If $\kappa_l<0$, $\Sigma$ is unstable.
\end{enumerate}
\end{proposition} 
\begin{proof}
We denote by $\Sigma_0\subset\mathcal{H}$ the cross-section of Definition \ref{d11} with $\theta_k>0$. Denoting the stability operator associated to $\Sigma_0$ by $\mathcal{L}_0$, with a given principal eigenfunction $\varphi_0>0$, and principal eigenvalue $\mu_0$, we observe upon integration
$$\mu_0\int_{\Sigma_0}\varphi_0dA = \int_{\Sigma_0}\mathcal{L}(\varphi_0)dA = \frac12\kappa_l\int_{\Sigma_0}\theta_k\varphi_0dA.$$
So $\mu_0$ takes the same sign as $\kappa_l$ and the claim follows for $\Sigma_0$. For an arbitrary cross-section $\Sigma=\{s=\omega\}$, $\omega\in\mathcal{F}(\Sigma_0)$, and function $\psi\in\mathcal{F}(\Sigma_0)$, we see from Lemma \ref{l3} that the stability operator takes the form:
\begin{align*}
\mathcal{L}(\psi\circ\pi) &=\Big( -\Delta_0\psi-2\nabla_0\cdot(\psi (t_0-\kappa_ld\omega))+\frac12\kappa_l\Big(\theta_k-2\Delta_0\omega-2\kappa_l|\nabla_0\omega|^2+4t_0(\nabla_0\omega)\Big)\psi\Big)\circ\pi\\
&=\Big(-\Delta_0\psi-2\nabla_0\cdot(\psi t_0)+\frac12\kappa_l\theta_k\psi+2\kappa_l\langle\nabla_0\psi,\nabla_0\omega\rangle-\psi e^{\kappa_l\omega}\Big(\Delta_0 e^{-\kappa_l\omega}+2t_0(\nabla_0 e^{-\kappa_l\omega})\Big)\Big)\circ\pi\\
&=e^{\kappa_l\omega}\Big(-(e^{-\kappa_l\omega}\Delta_0\psi+2\langle\nabla_0\psi,\nabla_0 e^{-\kappa_l\omega}\rangle+\psi\Delta_0 e^{-\kappa_l\omega})\\
&\qquad\qquad\qquad\qquad\qquad\qquad\qquad -2(e^{-\kappa_l\omega}\nabla_0\cdot(\psi t_0)+\psi t_0(\nabla_0 e^{-\kappa_l\omega}))+\frac12\kappa_l\theta_k\psi e^{-\kappa_l\omega}\Big)\circ\pi\\
&=e^{\kappa_l\omega}\Big(-\Delta_0(\psi e^{-\kappa_l\omega})-2\nabla_0\cdot(\psi e^{-\kappa_l\omega}t_0)+\frac12\kappa_l\theta_k\psi e^{-\kappa_l\omega}\Big)\circ\pi\\
&=e^{\kappa_l\omega}\mathcal{L}_0(\psi e^{-\kappa_l\omega})\circ\pi.
\end{align*}
It follows that $\varphi = (\varphi_0e^{\kappa_l\omega})\circ\pi$ is a principle eigenfunction for $\Sigma$ with principle eigenvalue $\mu=\mu_0$. We conclude that $\Sigma$ inherits the stability of $\Sigma_0$.
\end{proof}

\begin{definition}
We say an optically rigid WIH is strictly stable, marginally stable, or unstable whenever $\kappa_l>0$, $\kappa_l=0$, or $\kappa_l<0$ respectively.
\end{definition}
\begin{proposition}\label{p5}
A NEH, $\mathcal{H}$, admits only one strictly stable, optically rigid WIH structure.
\end{proposition}
\begin{proof}
We need to show whenever a NEH admits two optically rigid WIH structures, both of which are strictly stable with associated pregeodesic generators $l,l'$, then $l = cl'$ for some constant $c\neq0$. \\
\indent It is necessarily the case that $l'=fl$ for some $f\in\mathcal{H}$, $f\neq 0$. For any cross-section $\Sigma\subset\mathcal{H}$, the normal frame $\{l,k\}\subset\Gamma(T^\perp\Sigma)$ boosts to the normal frame $\{l',k'\}=\{fl,\frac{1}{f}k\}$. Relative to the frame $\{l',k'\}$, the  connection 1-form is given by $t'(v) = \frac12\langle D_v(k'),l'\rangle = \frac12\langle D_v(\frac{1}{f}k),fl\rangle = t(v)-v\log|f|$. Therefore, we observe from (\ref{e8}):
\begin{align*}
0&=\frac12\kappa_l l'\theta_{k'}\\
& = \kappa_l\Big(\frac12G(l',k')+\mathcal{K}+\nabla\cdot t'-|t'|^2-\frac12\kappa_{l'}\theta_{k'}\Big)\\
&=\kappa_l\Big(\frac12G(l,k)+\mathcal{K}+\nabla\cdot(t-d\log|f|)-|t-d\log|f||^2-\frac12\kappa_{l'}\frac{1}{f}\theta_k\Big)\\
&=\frac{\kappa_l}{f}\Big((\frac12G(l,k)+\mathcal{K}+\nabla\cdot t-|t|^2)f-\Delta f+2t(\nabla f)-\frac12\kappa_{l'}\theta_k\Big)\\
&=\frac{\kappa_l}{f}\Big(-\Delta f+2t(\nabla f)+\frac{1}{2}(\kappa_{l}f-\kappa_{l'})\theta_k\Big)\\
&=\frac{1}{f}\mathcal{L}^\star(\kappa_l f-\kappa_{l'}).
\end{align*} 
Since $\Sigma$ is strictly stable, $\mathcal{L}^\star$ has bounded inverse and we conclude that $f = \frac{\kappa_{l'}}{\kappa_l}$. Since $\Sigma$ was arbitrary, $f=\frac{\kappa_{l'}}{\kappa_l}$ throughout $\mathcal{H}$.
\end{proof}

\section{Existence and Stability of quasi-round MOTS}
\subsection{Foliation by quasi-round MOTS}
In this section we will show the first two of our main results. With the technology constructed in the previous section we are pretty much ready to show that a stable, optically rigid WIH admits a unique foliation by quasi-round MOTS. First we need two lemmata:
\begin{lemma}\label{l5}
For any cross-section $\Sigma:=\{s=\omega\}\subset\mathcal{H}$, $\omega\in\mathcal{F}(\Sigma_0)$, of a WIH, and $V,W\in E(\Sigma)$,
\begin{enumerate}
\item $\nabla_{\tilde V}\tilde W = (\nabla_s)_VW+((\nabla_s)_VW\omega) l$
\item $(\nabla_{\tilde V}t_\omega)(\tilde W) = ((\nabla_s)_Vt_s)(W) - \nabla_s^2(\kappa_l\omega)(V,W)$
\end{enumerate}
where $\nabla_s^2$ is the Hessian associated to $\nabla_s$ on $\Sigma_s$.
 \end{lemma}
 \begin{proof}
 For the first part of the proof it suffices to show that $\langle \nabla_{\tilde V}\tilde W,U\rangle = \langle (\nabla_s)_VW,U\rangle$ for any $U\in E(\Sigma)$:
 \begin{align*}
 \langle \nabla_{\tilde V}\tilde W, U\rangle&=\langle D_{\tilde V}\tilde W+\frac12\chi_k(\tilde V,\tilde W)l, U\rangle\\
 &=\langle D_{\tilde V}\tilde W, U\rangle\\
 &=\tilde V\langle W,U\rangle-\langle \tilde W, D_{\tilde V}U\rangle\\
 &=V\langle W,U\rangle - \langle W, \nabla_VU\rangle+\Big(V\omega l\langle W,U\rangle - W\omega\langle l,D_{\tilde V} U\rangle-V\omega\langle W,D_l U\rangle\Big)\\
 &=V\langle W,U\rangle - \langle W, D_VU\rangle\\
 &=\langle (\nabla_s)_VW,U\rangle
 \end{align*}
 where the third and final terms of the forth equality vanish due to $\chi_l =0$, and  the forth vanishes since $D$ restricts to a connection on $\mathcal{H}$.\\
For the second part of the lemma, we have from Lemma \ref{l3}, (\ref{e9}), and the result above:
\begin{align*}
(\nabla_{\tilde V}t_\omega)(\tilde W)&=\tilde Vt_\omega(\tilde W) - t_\omega(\nabla_{\tilde V}\tilde W)\\
&=(V+V\omega l)(t_s(W)-W(\kappa_l\omega)) - t_s(\nabla_VW)+(\nabla_s)_VW(\kappa_l\omega)\\
&=((\nabla_s)_V t_s)(W) - \nabla^2_s(\kappa_l\omega)(V,W).
\end{align*}
\end{proof}
\begin{lemma}\label{l6}
For any cross-section $\Sigma:=\{s=\omega\}\subset\mathcal{H}$, $\omega\in\mathcal{F}(\Sigma_0)$, of an optically rigid WIH, we have 
$$\frac{\kappa_l}{2}(\theta_k)_\omega = \Big(e^{-\kappa_l\omega}\mathcal{L}_0^\star(e^{\kappa_l\omega})\Big)\circ\pi.$$
\end{lemma}
\begin{proof}
From Lemma \ref{l3} and optical rigidity, we have 
\begin{align*}
\frac12\kappa_l(\theta_k)_\omega &=\Big(\frac12\kappa_l\theta_{k_0}-\Delta_0(\kappa_l\omega)-|\nabla_0(\kappa_l\omega)|^2+2t_0(\nabla_0(\kappa_l\omega))\Big)\circ\pi\\
&=e^{-\kappa_l\omega}\Big(\frac12\kappa_l\theta_{k_0}e^{\kappa_l\omega}-\Delta_0 e^{\kappa_l\omega}+2t_0(\nabla_0 e^{\kappa_l\omega})\Big)\circ\pi\\
&=\Big(e^{-\kappa_l\omega}\mathcal{L}_0^\star(e^{\kappa_l\omega})\Big)\circ\pi
\end{align*}
\end{proof}
\begin{theorem}\label{t1}
If $\mathcal{H}$ is a strictly stable, optically rigid WIH, then it admits a unique foliation by quasi-round MOTSs.
\end{theorem}
\begin{proof}
\underline{Existence}:\\
From standard results for the Laplace-Beltrami operator on compact Riemannian manifolds, the equation,
$$\Delta_0 u = \mathcal{K}_0+\nabla_0\cdot t_0-\frac{4\pi}{|\Sigma_0|}$$
is solvable on $\Sigma_0$ since the Gauss-Bonnet and Divergence theorems ensure both sides integrate to zero. From Elliptic Regularity and the Maximum Principle (see \cite{evans2010partial, gilbarg2015elliptic}) we also know that the solution $u$ is smooth and unique up to an additive constant respectively. Next, from the stability hypothesis on $\mathcal{H}$, the operator $\mathcal{L}_0^\star$ has bounded inverse so that we may solve for $\psi$ in the equation
$$\mathcal{L}_0^\star(\psi) = e^u.$$
Elliptic regularity once again ensures $\psi$ is smooth, and from Lemma \ref{l4}, $\psi>0$. Defining $\omega:=\frac{1}{\kappa_l}\log\psi$ we obtain a cross-section $\Sigma_\omega\hookrightarrow\mathcal{H}$ which by Lemma \ref{l6} satisfies $\frac12\kappa_le^{\kappa_l\omega\circ\pi}(\theta_k)_\omega = e^{u\circ\pi}$. We conclude that $(\theta_k)_\omega>0$ and therefore Lemma \ref{l6} coupled with optical rigidity gives:
\begin{align*}
\rho_\omega &=\mathcal{K}+\nabla_\omega\cdot t_\omega - \Delta_\omega\log(\theta_k)_\omega\\
&=\Big(\mathcal{K}_0+\nabla_0\cdot (t_0-d(\kappa_l\omega))-\Delta_0\log(\frac{2}{\kappa_l}e^{u-\kappa_l\omega})\Big)\circ\pi\\
&=\frac{4\pi}{|\Sigma_0|}.
\end{align*}
We also observe that $\rho_{\omega+s} = \rho_\omega$ for any constant $s$ by optical rigidity of $\mathcal{H}$.\\
\underline{Uniqueness}:\\
From the expression for $\rho_\omega$ above, in order for $\Sigma_{\omega'}$ to satisfy $\rho_{\omega'} = \frac{4\pi}{|\Sigma|}$, it necessarily has to be the case that 
$$\Big(\frac{4\pi}{|\Sigma_0|}-\mathcal{K}_0-\nabla_0\cdot t_0\Big)\circ\pi = -\Delta_0(\kappa_l\omega'\circ\pi)-\Delta_0\log(\theta_k)_{\omega'}=-\Delta_0\log(e^{\kappa_l\omega'\circ\pi}(\theta_k)_{\omega'}).$$ 
From Lemma \ref{l6} we therefore conclude that $e^{u+C}= \mathcal{L}_0^\star(\frac{2}{\kappa_l}e^{\kappa_l\omega'})$ for some constant $C$. Since $\mathcal{L}^\star$ has bounded inverse we have $\frac{2}{\kappa_l}e^{\kappa_l\omega'} = \psi e^C = e^{\kappa_l\omega+C}$ and therefore
$$\omega' = \omega+\frac{C+\log\kappa_l-\log2}{\kappa_l}.$$
We see that $\Sigma_{\omega'}$ is simply a translate of $\Sigma_\omega$ along $l$, moreover, as the constant $C$ runs through values of $\mathbb{R}$ we recover the foliation of the existence argument.
\end{proof}
\subsection{Stability}
In this section, we will assume our MOTS $\Sigma_0$ satisfies the necessary conditions allowing the construction of a Null Inflation Basis. Consequently, we will henceforth take $\mathcal{L}$ to be the stability operator along $L^-$ relative to the Null Inflation basis:
 $$\mathcal{L}(\psi) = -\Delta\psi -2\nabla\cdot(\psi\tau)+\Big(\rho_0-|\tau|^2+\frac12G(L^+,L^-)\Big)\psi$$
 where $\rho_0 = \mathcal{K}+\nabla\cdot\tau$.\newpage
\begin{proposition}\label{p6} Given an embedded surface $S\subset\mathcal{M}$, admitting a Null Inflation basis $\{L^-,L^+\}\subset\Gamma(T^\perp S)$, the following holds:
\begin{align}
\delta_{\psi L^-}\mathcal{K} &= -\psi\mathcal{K}-\frac12\Delta\psi+\nabla\cdot\nabla\cdot(\psi\hat{\chi}^-)\label{e12}\\
\delta_{\psi L^+}\mathcal{K} &=-\psi\langle\vec{H},\vec{H}\rangle\mathcal{K}-\frac12\Delta(\psi\langle\vec{H},\vec{H}\rangle)+\nabla\cdot\nabla\cdot(\psi\hat\chi^+)\label{e13}\\
\delta_{\psi L^-}\langle\vec{H},\vec{H}\rangle &= 2\mathcal{L}(\psi)-\langle\vec{H},\vec{H}\rangle\Big(\frac32+|\hat{\chi}^-|^2+G(L^-,L^-)\Big)\psi\label{e14}\\
\delta_{\psi L^+}\langle\vec{H},\vec{H}\rangle &= 2\langle\vec{H},\vec{H}\rangle\mathcal{L}^\star(\psi)-\Big(\frac32\langle\vec{H},\vec{H}\rangle^2+|\hat{\chi}^+|^2+G(L^+,L^+)\Big)\psi\label{e15}\\
\delta_{\psi L^-}\tau&=-\psi\tau-\nabla\cdot(\psi \hat{\chi}^-)+\psi G_{L^-}+d(\psi(|\hat{\chi}^-|^2+G(L^-,L^-)))\label{e16}\\
\delta_{\psi L^+}\tau &= -2d\mathcal{L}^\star(\psi)-\psi G_{L^+}+\nabla\cdot(\psi \hat{\chi}^+)-\psi\langle\vec{H},\vec{H}\rangle\tau+\frac12\Big(\langle\vec{H},\vec{H}\rangle d\psi-\psi d\langle\vec{H},\vec{H}\rangle\Big)\label{e17}
\end{align}
\end{proposition}
\begin{proof}
Firstly, we start by considering any neighborhood where $\psi\neq0$. The result follows for (\ref{e12}) directly from the structure equations by setting $L = \psi L^+,\,\, \ubar L = \frac{1}{\psi}L^-$, similarly for (\ref{e13}) by setting $L=\frac{1}{\psi}L^-,\,\,\ubar L = \psi L^+$. Toward showing (\ref{e14},\ref{e15}) we calculate
\begin{align*}
\ubar L\langle\vec{H},\vec{H}\rangle &=\ubar L(\tr\ubar\chi\tr\chi)\\
&=\Big(-\frac12\tr\ubar\chi^2-|\hat{\ubar\chi}|^2-G(\ubar L,\ubar L)+\kappa\tr\ubar\chi\Big)\tr\chi\\
&\qquad\qquad\qquad\qquad\qquad\qquad\qquad+\tr\ubar\chi\Big(G(L,\ubar L)+2\mathcal{K}-2\nabla\cdot\zeta-2|\zeta|^2-\langle\vec{H},\vec{H}\rangle-\kappa\tr\chi\Big)\\
&=-\frac32\tr\ubar\chi\langle\vec{H},\vec{H}\rangle-\tr\chi\Big(|\hat{\ubar\chi}|^2+G(\ubar L,\ubar L)\Big)+2\tr\ubar\chi\Big(\frac12G(L^-,L^+)+\mathcal{K}-\nabla\cdot\zeta-|\zeta|^2\Big).
\end{align*}
Setting $\ubar L = \psi L^{\mp},\,\,L = \frac{1}{\psi}L^{\pm}$, whereby $\zeta = d\log|\psi|\pm\tau$, we see
\begin{align*}
\psi&\Big(\frac12G(L^-,L^+)+\mathcal{K}-\nabla\cdot\zeta-|\zeta|^2\Big) \\
&= \psi\Big(\frac12G(L^-,L^+)+\rho_0-\nabla\cdot\tau-\Delta\log|\psi|\mp\nabla\cdot\tau-|\tau|^2\mp2\tau(\nabla\log|\psi|)-|\nabla\log|\psi||^2\Big)\\
&= \psi\Big(\frac12G(L^-,L^+)+\rho_0-|\tau|^2\Big)-\Delta\psi\mp2\tau(\nabla\psi)+(\mp1-1)\psi\nabla\cdot\tau\\
&=\frac12(1\pm1)\mathcal{L}(\psi)+\frac12(1\mp1)\mathcal{L}^\star(\psi).
\end{align*}
So we conclude with (\ref{e14}) by setting $\ubar L = \psi L^-$, where $\tr\ubar\chi = \psi$, and (\ref{e15}) by setting $\ubar L = \psi L^+$, where $\tr\ubar\chi=\psi\langle\vec{H},\vec{H}\rangle$.
To show (\ref{e16}) we calculate
\begin{align*}
\mathcal{L}_{\ubar L}\tau &=\mathcal{L}_{\ubar L}\zeta-d\frac{1}{\tr\ubar\chi}\Big(-\frac12\tr\ubar\chi^2-|\hat{\ubar\chi}|^2-G(\ubar L,\ubar L)+\kappa\tr\ubar\chi\Big)\\
&=G_{\ubar L}-\nabla\cdot\hat{\ubar\chi}-\tr\ubar\chi\tau+d\Big(\frac{1}{\tr\ubar\chi}(|\hat{\ubar\chi}|^2+G(\ubar L,\ubar L))\Big)
\end{align*}
and the result follows for $\ubar L = \psi L^-$. For (\ref{e17}) we observe, by switching the roles of $\ubar L$ and $L$ in the structure equations, that
\begin{align*}
\mathcal{L}_L\zeta &= -G_L+\nabla\cdot\hat{\chi}-\tr\chi\zeta-\frac12 d\tr\chi-d\kappa\\
L\log\tr\ubar\chi&=\frac{1}{\tr\ubar\chi}\Big(G(L,\ubar L)+2\mathcal{K}+2\nabla\cdot\zeta-2|\zeta|^2\Big)-\tr\chi-\kappa,
\end{align*}
giving,
\begin{align*}
\mathcal{L}_L\tau&=-G_L+\nabla\cdot\hat{\chi}-\tr\chi\zeta+\frac12 d\tr\chi-d\Big(\frac{1}{\tr\ubar\chi}\Big(G(L,\ubar L)+2\mathcal{K}+2\nabla\cdot\zeta-2|\zeta|^2\Big)\Big).
\end{align*}
The result follows by setting $L = \psi L^+,\,\, \ubar L = \frac{1}{\psi}L^-$ whereby $\zeta = \tau-d\log|\psi|$ and recalling our calculations for (\ref{e14},\ref{e15}).\\
\indent In any neighborhood where $\psi$ vanishes identically (\ref{e12})-(\ref{e17}) holds, since by construction, the geometry remains invariant. The remaining possibilities are settled by continuity of both the left and right sides of the equality in (\ref{e12})-(\ref{e17}).
\end{proof}
For the remainder of the paper we define the H\"{o}lder space $\ring{C}^{k,\alpha}(\Sigma)\subset C^{k,\alpha}(\Sigma)$ of functions such that $f\in \ring{C}^{k,\alpha}(\Sigma)$ is characterized by the property that $\int_\Sigma fdA = 0$.
\begin{theorem}\label{t2}
Given a quasi-round MOTS, $\Sigma$, such that $\delta_{L^+}\langle\vec{H},\vec{H}\rangle=0$, the following linearization holds:
$$\delta_{\psi L^++\phi L^-}
\begin{pmatrix}
\rho_0-\frac{4\pi}{|\Sigma|}\\
\langle\vec{H},\vec{H}\rangle
\end{pmatrix}
=
\begin{pmatrix}
-2\Delta\mathcal{L}^\star&G\\
0&2\mathcal{L}
\end{pmatrix}
\begin{pmatrix}
\psi\\\phi
\end{pmatrix}$$
whereby $G(\phi) = \frac{4\pi}{|\Sigma|}(\frac{\int\phi}{|\Sigma|}-\phi)+\nabla\cdot\Big(\phi(G_{L^-}-2\chi^-\circ\tau)\Big)+\Delta\Big(\phi(|\hat{\chi}^-|^2+G(L^-,L^-)-\frac12)\Big)$, and $\mathcal{L}$ is the stability operator along $L^-$. Moreover, if $\Sigma$ is strictly stable, then the linearization has bounded inverse on $\ring{C}^{k,\alpha}(\Sigma)\times C^{l,\beta}(\Sigma)$.
\end{theorem}
\begin{proof}
From the hypotheses on $\Sigma$ and (\ref{e15}), it follows that $\chi^+ = 0$ and $G(L^+,L^+)=0$, and from the Dominant Energy Condition that $G_{L^+}=0$. Therefore, the second row of the matrix representation for the linearization follows from (\ref{e14}) and (\ref{e15}) of Proposition \ref{p6}. By the first variation of area formula, we have $\delta_{\psi L^+} |\Sigma| =\int_\Sigma-\langle\vec{H},\psi L^+\rangle dA= \int_\Sigma \psi\langle\vec{H},\vec{H}\rangle dA = 0$. Therefore, using (\ref{e13}) of Proposition \ref{p6}, the first entry of the first row satisfies $\delta_{\psi L^+}\rho_0 = \delta_{\psi L^+}(\nabla\cdot\tau)$. It's a standard exercise (see, for example, \cite{R} Corollary 3.1.1) to verify that $\ubar L\nabla\cdot\tau = \tr\ubar\chi\nabla\cdot\tau-2\nabla\cdot(\hat{\ubar\chi}\circ\tau)+\nabla\cdot(\pounds_{\ubar L}\tau)$, which, for $\ubar L = \psi L^+$, gives us $\delta_{\psi L^+}\rho_0 = \nabla\cdot(\delta_{\psi L^+}\tau) = -2\Delta\mathcal{L}(\psi)$ from (\ref{e17}) of Proposition \ref{p6}. Once again, by the first variation of area formula, we have $\delta_{\phi L^-}|\Sigma| = \int_\Sigma-\langle\vec{H},\phi L^-\rangle dA = \int_\Sigma\phi dA$. The formula for $G(\phi)$ therefore follows from the formula for $\ubar L\nabla\cdot\tau$, (\ref{e12}), and (\ref{e16}).\\
\indent For the second part of our Theorem, since all operators have smooth coefficients on a compact manifold, it is a standard argument using a partition of unity to locally reduce to an operator on $\mathbb{R}^2$ (see, for example, \cite{evans2010partial, aubin2012nonlinear}) from which it follows that the linearization is a bounded operator. It suffices therefore, by way of the Bounded Inverse Theorem, to show that the linearization is a bijection. Since the linearization is upper diagonal, this in turn is equivalent to showing the operators along the diagonal of the linearization are bijective in their respective Banach spaces. By the hypothesis that $\Sigma$ be strictly stable both $\mathcal{L}$ and $\mathcal{L}^\star$ have bounded inverses so it remains to show $\Delta\mathcal{L}^\star:\ring{C}^{k,\alpha}(\Sigma)\to\ring{C}^{k-4,\alpha}(\Sigma)$ is bijective.\\
\underline{Injectivity}\\
Given $\psi\in\ring{C}^{k+4,\alpha}(\Sigma)$, $\Delta\mathcal{L}^\star(\psi) = 0$ necessitates that $\mathcal{L}^\star(\psi) = C$ for some constant $C$. Given the case that $C=0$, $\psi=0$ follows by the existence of a bounded inverse for $\mathcal{L}^\star$. On the other hand, if it happens that $C\neq 0$, then $\mathcal{L}^\star(\frac{\psi}{C})>0$ implies that $\frac{\psi}{C}>0$ by Lemma \ref{l4} and therefore $\int_\Sigma\frac{\psi}{C}dA>0$. This contradicts the fact that $\int\psi dA=0$, so $\psi =0
$ and we conclude that $\Delta\mathcal{L}^\star$ is injective.\\
\underline{Surjectivity}\\
It's a well known fact that $\Delta:C^{k,\alpha}(\Sigma)\to\ring{C}^{k-2,\alpha}(\Sigma)$ is surjective so it suffices to show, for each $u\in C^{k,\alpha}(\Sigma)$, the existence of a constant $C_u$ and $v\in\ring{C}^{k+2,\alpha}(\Sigma)$ such that $\mathcal{L}^\star{v}=u+C_u$. Once again using Lemma \ref{l4} and the fact that $\mathcal{L}^\star$ has bounded inverse, we find unique $\psi_1>0, \psi_u\in C^{k+2,\alpha}(\Sigma)$ such that $\mathcal{L}^\star(\psi_u) = u$ and $\mathcal{L}^\star(\psi_1) = 1$. The desired function is therefore given by $v = \psi_u-\frac{\int \psi_u}{\int \psi_1}\psi_1$ whereby $C_u = -\frac{\int\psi_u}{\int\psi_1}$.
\end{proof}
For the final result of this section we will need to construct a convenient coordinate system in a neighborhood of a 2-sphere. The following result is an adaptation of the more general result found in \cite{andersson2008} (Lemma 6.1):
\begin{lemma}\label{l2.2.1}
Given an embedded 2-sphere $\Sigma\hookrightarrow\mathcal{M}$, there exists a spacetime neighborhood $\mathcal{V}$ of $\Sigma$, with local coordinates $(t,r,x^i)$ on $\mathcal{V}$ and functions $Z, \vartheta, \eta^i, h_{ij}$ such that the metric takes the form
$$g = e^Z(dt\otimes dr+dr\otimes dt)+h_{ij}(dx^i-\eta^idr)\otimes(dx^j-\eta^jdr)$$
where $\Sigma\cap\mathcal{V}=\{t=0,r=0\}$, $Z(t=0, r=0, x^i)=\log 2$, $\eta^i(t=0, r=0, x^i)=0$, and $h_{ij}$ is a positive definite 2-matrix.
\end{lemma}
\begin{proof}
We start by choosing a null basis $\{L,\ubar L\}\subset\Gamma(T^\perp\Sigma)$, $\ubar L$ past-pointing. For sufficiently small $|t|$, the map $(p,t)\to \exp(tL|_p)$, $p\in\Sigma$, defines a smooth embedding of a null hypersurface $\mathcal{N}\hookrightarrow\mathcal{M}$ with corresponding foliation $\{\Sigma_t\}\subset\mathcal{N}$ whereby $\Sigma_0=\Sigma$ and each $\Sigma_t$ is a spacelike 2-sphere. If we denote the null tangent along $\mathcal{N}$ also by $L$ then each $\Sigma_t$ admits an adapted null normal $\ubar L_t$ such that $\langle L,\ubar L_t\rangle=2$, moreover, $\ubar L_0 = \ubar L$. Standard smooth dependence on initial conditions of ODEs, coupled with the inverse function theorem then guarantees that if we collect null geodesics along $\ubar L_t$, we fill-in a neighborhood $\mathcal{V}$ of $\Sigma$. $\mathcal{V}$ is foliated by smooth null hypersurfaces $\{S_t\}$ off of $\ubar L_t\in\Gamma(T^\perp\Sigma_t)$, whereby $\mathcal{N}\cap S_t = \Sigma_t$. By shrinking $\mathcal{V}$ if necessary, the parameter $t$ extends to a smooth function whereby $S_{t_0}=\{t=t_0\}$, and $\text{grad}(t)\in\Gamma(T^\perp S_t)\subset\Gamma(TS_t)$, giving $|\text{grad}(t)|^2=0$. From the identity $D_{\text{grad}(t)}\text{grad}(t) = \frac12\text{grad}|\text{grad}(t)|^2=0$ it follows that $\text{grad}(t)$ generates null geodesics ruling the leaves of the foliation $\{S_t\}$, and the vector field $2\text{grad}(t)$ extends $\ubar L_t$ off of $\mathcal{N}$ to all of $\mathcal{V}$.\\
\indent For sufficiently small $|r|$ the map $(p,r)\to\exp(2r\nabla t|_p)$, $p\in\Sigma$, induces a foliation $\{\Sigma_r\}\subset S_0$ with adapted null basis $\{L_r,2\text{grad}(t)\}$. Repeating the process above we obtain another smooth foliation of $\mathcal{V}$ by null hypersurfaces $\{T_r\}$, generated by the similarly constructed null geodesic vector field $\text{grad}(r)$. We now simply carry local co-ordinate functions $(x^1,x^2)$ from $\Sigma$ to $S_{t_0}\cap T_{r_0}$ by Lie-dragging $x^i$ along $\text{grad}(t)$ from $\Sigma$ to $\Sigma_{r_0}$, and then along $\text{grad}(r)$ to $S_{t_0}\cap T_{r_0}$. This construction gives $t$-coordinate curves that are null pregeodesic, $\partial_t\propto\text{grad}(r)$, so that
$$g = e^{Z}(dt\otimes dr+dr\otimes dt)+\vartheta dr\otimes dr+h_{ij}(dx^i-\eta^i dr)\otimes(dx^j-\eta^j dr).$$
On $\Sigma$ we have $L=\partial_t$, $\ubar L = \partial_r$, $Z\equiv \log 2$, $\vartheta = 0$, $\eta\equiv 0$, and $h_{ij} = \gamma_{ij}$.\\
\indent Since $\text{grad}(t)$ is null, we have $\text{grad}(t)(t) = |\text{grad}(t)|^2 = 0$, so that $\text{grad}(t) = a\partial_r+\alpha^i\partial_i$. Moreover, $a\neq 0$ since $\vec{\alpha}:=\alpha^i\partial_i$ is spacelike. From this we see that $0=\partial_i(t) = \langle \partial_i,\text{grad}(t)\rangle = -a\eta_i+\alpha_i$ whereby $\eta_i = h_{ij}\eta^j$, $\alpha_i = h_{ij}\alpha^i$. Moreover, $0=\partial_r(t) = \langle \partial_r, \text{grad}(t)\rangle = a\vartheta + a|\vec{\eta}|^2-\vec{\eta}\cdot\vec{\alpha}=a\vartheta$ where $\vec{\eta}:=\eta^i\partial_i$, giving $\vartheta=0$.
\end{proof}
\begin{theorem}\label{t3}
Consider a neighborhood $\mathcal{V}:=(-a,a)\times(-b,b)\times\mathbb{S}^2$, $a,b>0$ and a smooth Lorentzian metric $g_0\in \text{Sym}(T^\star\mathcal{V}\otimes T^\star\mathcal{V})$ satisfying the Dominant Energy Condition. Assume also that $\Sigma_0\cong \{0\}\times\{0\}\times\mathbb{S}^2$ is a strictly stable quasi-round MOTS in the geometry $(\mathcal{V},g_0)$ satisfying $\delta_{L^+}\langle\vec{H},\vec{H}\rangle=0$. Then, given a smooth path of Lorentzian metrics $\lambda\to g_\lambda\in \text{Sym}(T^\star\mathcal{V}\otimes T^\star\mathcal{V})$, $\lambda\in [0,c)\subset\mathbb{R}$, each $g_\lambda$ satisfying the Dominant Energy Condition, there exists an $0<\epsilon<c$ and a unique smooth path of embeddings, $\phi_\lambda:\mathbb{S}^2\hookrightarrow\mathcal{V}$, $0\leq \lambda\leq \epsilon$, such that each $\Sigma_\lambda:=\phi_\lambda(\mathbb{S}^2)$ is a quasi-round MOTS in the geometry $(\mathcal{V},g_\lambda)$.
\end{theorem}
\begin{proof}
We take $\epsilon>0$ sufficiently small to ensure the induced metric $h_\lambda = g_\lambda|_{\Sigma_0}$ remains positive definite for $0\leq \lambda\leq \epsilon$. We also choose a smoothly varying null normal $\ubar L_\lambda\in\Gamma(T^\perp\Sigma_0)_{g_\lambda}$ and shrink $\epsilon$ if necessary so that $\tr\ubar\chi_\lambda:=\ubar L_\lambda\log\sqrt{\det h_{ij}(\lambda)}>0$, giving a smoothly varying Null Inflation Basis $\{L^-_\lambda, L^+_\lambda\}\subset\Gamma(T^\perp\Sigma_0)_{g_\lambda}$. Since the co-ordinates in Lemma \ref{l2.2.1} depend smoothly on the metric and choice of normal null basis, we apply the construction using $\{L^-_\lambda,L^+_\lambda\}$ and conclude that a sufficiently small neighborhood $\mathcal{V}_0\subset\mathcal{V}$ exists on which all metrics take the form
\begin{align*}
g_\lambda &= e^Z(dt\otimes dr+dr\otimes dt)+ h_{ij}(dx^i-\eta^i dr)\otimes(dx^j-\eta^jdr)\\
g^{-1}_\lambda&=e^{-Z}\Big(\partial_t\otimes(\partial_r+\eta^i\partial_i)+(\partial_r+\eta^i\partial_i)\otimes\partial_t\Big)+h^{ij}\partial_i\otimes\partial_j
\end{align*}
where $\Sigma_0\cap\mathcal{V}_0=\{t=r=0\}$, $Z(\lambda, t=0,r=0,x^i)=\log2$, $\eta^i(\lambda, t=0,r=0,x^i)=0$, and $h_{ij}(\lambda,t,r,x^i)$ is positive definite.\\ 
\indent We now choose sufficiently small $C$ such that any $f,g\in C^{k+4,\alpha}(\Sigma_0)$ satisfying $|f|_{k+4,\alpha,\Sigma_0},|g|_{k+4,\alpha,\Sigma_0}\leq C$ ensures $(t=f(x_i),r=g(x_i),x_i)\in\mathcal{V}_0$. This defines an embedding $\Phi(f,g)(\Sigma_0):=\Sigma_{f,g}$ with induced metric (suppressing $\lambda$ dependence) $\gamma_{f,g} := (h_{ij}+A_{ij})dx^i\otimes dx^j$ whereby
$$A_{ij}=e^{Z}(f_ig_j+f_jg_i)+|\eta|^2g_ig_j-(\eta_ig_j+\eta_jg_i)\in C^{k+3,\alpha}(\Sigma).$$
Defining $\tr A:=h^{ij}A_{ij}$, $\hat{A} := A-\frac12(\tr A)h$, we leave it to the reader to verify, by shrinking $C$ if necessary, one ensures $(1+\frac12\tr A)>\frac{1}{\sqrt{2}}|\hat{A}|$ for all $\lambda\leq \epsilon$, consequently:
$$\gamma^{-1}_{f,g} = \frac{(1+\frac12\tr A)h^{ij}-\hat{A}^{ij}}{(1+\frac12\tr A)^2-\frac12|\hat{A}|^2}\partial_i\otimes\partial_j.$$
It follows that $g_\lambda|_{T^\perp\Sigma_{f,g}}$ is non-degenerate and $T^\perp\Sigma_{f,g}$ is trivial with basis vectors
\begin{align*}
N_1&:=\text{grad}(t-f)=-e^{-Z}\vec{\eta}\cdot\vec{\nabla}f\partial_t+e^{-Z}(\partial_r+\vec{\eta})-\vec{\nabla} f,\\
N_2&:=\text{grad}(r-g)=e^{-Z}(1-\vec{\eta}\cdot\vec{\nabla}g)\partial_t-\vec{\nabla}g,
\end{align*}
where we define, $\vec{\nabla}f:=h^{ij}\partial_if\partial_j$, when restricted to $\Sigma_{f,g}$. We have (with a slight abuse of notation) $d\Phi(f,g)(\partial_i) = \partial_i+f_i\partial_t+g_i\partial_r$ and conclude that 
\begin{align*}
D_{d\Phi(\partial_i)}d\Phi(\partial_j) &= f_{ij}\partial_t+g_{ij}\partial_r\\
&+f_jD_{\partial_i}\partial_t+g_jD_{\partial_i}\partial_r+f_if_jD_{\partial_t}\partial_t+f_ig_jD_{\partial_t}\partial_r+g_if_jD_{\partial_r}\partial_t+g_ig_jD_{\partial_r}\partial_r\\
&+D_{\partial_i}\partial_j.
\end{align*}
Since $\langle \partial_t,N_1\rangle = \partial_t(t-f)=1$, $\langle\partial_t,N_2\rangle=\partial_t(r-g) = 0$, $\langle \partial_r,N_1\rangle = \partial_r(t-f)=0$, $\langle\partial_r,N_2\rangle=\partial_r(r-g)=1$ we have
\begin{align*}
\langle\vec{H},N_1\rangle&=\gamma^{ij}\langle D_{\Phi(\partial_i)}\Phi(\partial_j),N_1\rangle = \gamma^{ij} f_{ij}+F_1(\partial f,\partial g,f,g)=:\mathcal{L}_1(f)\\
\langle\vec{H},N_2\rangle&=\gamma^{ij} g_{ij}+F_2(\partial f,\partial g,f,g)=:\mathcal{L}_2(g)
\end{align*}
with $F_i$ smooth functions. From the non-degeneracy of $g_\lambda|_{T^\perp\Sigma_{f,g}}$ we conclude $\mathcal{D}^2:=\langle N_1,N_2\rangle^2-\langle N_1,N_1\rangle\langle N_2,N_2\rangle>0$ so that $\mathfrak{a},\mathfrak{b}\in C^{k+3,\alpha}(\Sigma)$ given by
$$\mathfrak{a}: = \frac{-\langle N_1,N_1\rangle}{\mathcal{D}+\langle N_1,N_2\rangle},\,\,\mathfrak{b}=\frac{1}{\langle N_1,N_2\rangle+a\langle N_2,N_2\rangle}$$
produces the null vector fields $\ubar L_{f,g}:=2N_1+2\mathfrak{a}N_2,\,\,L_{f,g}:=\mathfrak{b}N_2-\frac12\langle N_2,N_2\rangle \mathfrak{b}^2\ubar L_{f,g}\in\Gamma(T^\perp\Sigma_{f,g})$ satisfying $\langle \ubar L,L\rangle = 2$, since $\mathfrak{b}\langle \ubar L_{f,g},N_2\rangle=2$. It's easily verified for $f=g=0$ that $\mathfrak{a}=0$, $\mathfrak{b}=2$ giving $\ubar L_{0,0} =\partial_r=L^-_\lambda$, $L_{0,0}=\partial_t = L^+_\lambda$ on $\Sigma$. Moreover, we conclude that the expansions along $\ubar L_{f,g}, L_{f,g}$ are given by
$$\begin{pmatrix}
\tr\ubar\chi_{f,g}\\\tr\chi_{f,g}
\end{pmatrix}
=-\begin{pmatrix}
2&2\mathfrak{a}\\
-\langle N_2,N_2\rangle \mathfrak{b}^2&\mathfrak{b}(1-\langle N_2,N_2\rangle \mathfrak{a}\mathfrak{b})
\end{pmatrix}
\begin{pmatrix}
\mathcal{L}_1(f)\\
\mathcal{L}_2(g)
\end{pmatrix}
$$
for each $\lambda$. Since $\tr\ubar\chi(\lambda, t=0,r=0,x^i)=1$, we shrink $C$ if necessary so that $\tr\ubar\chi(\lambda, f(x^i),g(x^i),x^i)>0$. Thus, for each fixed $\lambda\leq\epsilon$, the map $\ring{C}^{k+4,\alpha}(\Sigma_0)\times C^{k+4,\alpha}(\Sigma_0)\to\ring{C}^{k,\alpha}(\Sigma_0)\times C^{k+2,\alpha}(\Sigma_0)$ given by 
$$\begin{pmatrix}
f\\g
\end{pmatrix}
\to
\begin{pmatrix}
\mathcal{K}_{\Sigma_{f,g}}+\nabla\cdot\zeta_{f,g}-\Delta\log\tr\ubar\chi_{f,g}-\frac{4\pi}{|\Sigma_{f,g}|}\\\tr\ubar\chi_{f,g}\tr\chi_{f,g}
\end{pmatrix}_\lambda$$
is well defined. We recognize for $\lambda = 0$ that the linearization of this map at $(f,g)\equiv0$ is given in Theorem \ref{t2} with bounded inverse, therefore satisfying the hypotheses of the Banach Space Implicit Function Theorem. By shrinking $\epsilon>0$ if necessary, we therefore conclude with the unique existence of some $(f_\lambda, g_\lambda)\in \ring{C}^{k+4,\alpha}(\Sigma_0)\times C^{k+4,\alpha}(\Sigma_0)$ for each $\lambda\leq\epsilon$ as desired in the statement of our Theorem (see the figure below).  From the induced metric on $\Sigma_{f_\lambda,g_\lambda}$ and the expressions of $N_i$ we conclude that
$$-\Delta_{h_\lambda}\log\tr\ubar\chi_{f_\lambda,g_\lambda} = \frac{4\pi}{|\Sigma_{f_\lambda,g_\lambda}|}-\mathcal{K}_{\Sigma_{f_\lambda,g_\lambda}}-\nabla\cdot\zeta_{f_\lambda,g_\lambda}\in C^{k+1,\alpha}(\Sigma_0).$$ From standard regularity results for second order elliptic PDE (see, for example \cite{gilbarg2015elliptic}) we conclude that $\tr\ubar\chi_{f_\lambda,g_\lambda}\in C^{k+3,\alpha}(\Sigma_0)$ and therefore
$$\begin{pmatrix}
\mathcal{L}_1(f_\lambda)\\
\mathcal{L}_2(g_\lambda)
\end{pmatrix}
=\begin{pmatrix}
\frac12(\langle N_2,N_2\rangle \mathfrak{a}\mathfrak{b}-1)&\frac{\mathfrak{a}}{\mathfrak{b}}\\
-\frac12\langle N_2,N_2\rangle \mathfrak{b}&-\frac{1}{\mathfrak{b}}
\end{pmatrix}
\begin{pmatrix}
\tr\ubar\chi_{f_\lambda,g_\lambda}\\0
\end{pmatrix}
\in C^{k+3,\alpha}(\Sigma_0)\times C^{k+3,\alpha}(\Sigma_0)
.$$
It follows that $f_\lambda,g_\lambda \in C^{k+5,\alpha}(\Sigma_0)$, itterating the above procedure yields $f_\lambda, g_\lambda\in C^{\infty}(\Sigma_0)$.
\end{proof}

\begin{figure}[h]
\centering
\includegraphics[width=0.45\textwidth]{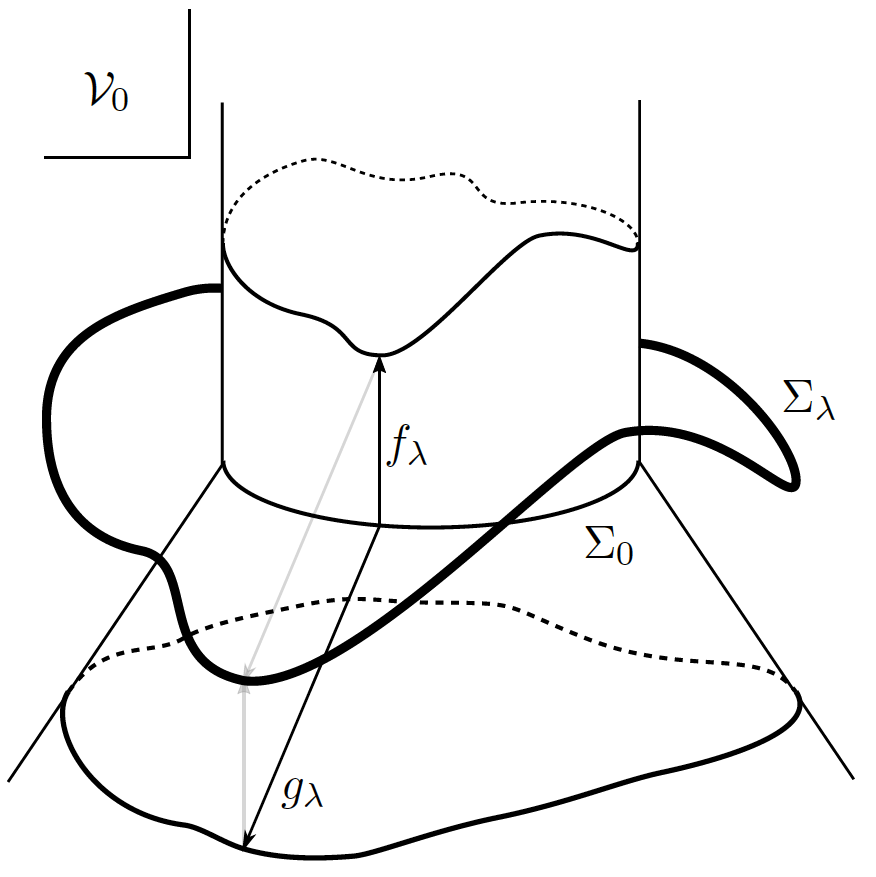}
\end{figure}

\section{Stability of the Schwarzschild Null Penrose Inequality}
\subsection{Schwarzschild Geometry, $\mathcal{H}_S$, and $\Omega_S$}
Schwarzschild spacetime models a static black hole of mass $M>0$. The maximal extension of this geometry is called the Kruskal spacetime $(\mathbb{P}\times_r\mathbb{S}^2,g_K)$ which is given by the warped product of the Kruskal Plane $\mathbb{P}:=\{(\frak{u},\frak{v})|\frak{u}\frak{v}>-2Me^{-1}\}$ and the standard round $\mathbb{S}^2$ with warping function $r=g^{-1}(\frak{u}\frak{v})$ for $g(r) = (r-2M)e^{\frac{r}{2M}-1}$, $r>0$. The metric and its inverse are given by:
\begin{align*}
g_K &= F(r)(d\frak{u}\otimes d\frak{v}+d\frak{v}\otimes d\frak{u}) + r^2(d\vartheta\otimes d\vartheta+(\sin\vartheta)^2 d\varphi\otimes d\varphi)\\
g_K^{-1}&=\frac{1}{F}(\partial_\frak{v}\otimes\partial_\frak{u} +\partial_\frak{u}\otimes\partial_\frak{v})+r^{-2}(\partial_\vartheta\otimes\partial_\vartheta+(\sin\vartheta)^{-2}\partial_\varphi\otimes\partial_\varphi)
\end{align*}
where $F(r) = \frac{8M^2}{r}e^{1-\frac{r}{2M}}$, and here $(\vartheta,\varphi)$ represent standard coordinates on $\mathbb{S}^2$. Each round $\mathbb{S}^2$ has area $4\pi r^2$ so we interpret $r$ as a `radius' function and $F(r)$ gives rise to unbounded curvature at $r=0$ and the `black hole' singularity.\\
\indent We observe a NEH $\mathcal{H}_S:=\{\frak{u}=0\}$, where the metric restricts to the degenerate metric
$$\gamma = g|_{\mathcal{H}_S} = (2M)^2(d\vartheta\otimes d\vartheta+(\sin\vartheta)^2 d\varphi\otimes d\varphi).$$
Clearly $\mathcal{H}_S$ is a NEH since any cross-section is homothetic to the standard round sphere. We will concentrate on the physically relevant black-hole geometry characterized by the region $\{\frak{v}>0\}$ where $\mathcal{H}_S$ separates the black hole exterior $\{\frak{u}>0,\frak{v}>0\}$ from the interior $\{\frak{u}<0,\frak{v}>0\}$. To show that $\mathcal{H}_S$ is a Killing horizon it will be convenient to make the co-ordinate change $(\frak{u},\frak{v})\to (v,r)$ whereby $v := 4M\log\frak{v}$. This puts the black hole geometry in the so called, \textit{Ingoing Eddington-Finkelstein} co-ordinate chart:
$$g = -(1-\frac{2M}{r})dv\otimes dv+(dr\otimes dv+dv\otimes dr)+ r^2(d\vartheta\otimes d\vartheta+(\sin\vartheta)^2 d\varphi\otimes d\varphi).$$
Since all the metric components are independent of the $v$ coordinate, we conclude that $\partial_v$ is a Killing vector field. Moreover, $\{\frak{u}=0\}=\{r=2M\}$, therefore $\partial_v$ restricts to a null pre-geodesic generator on $\mathcal{H}_S$, therefore $\mathcal{H}_S$ is a Killing Horizon. We also readily calculate the surface gravity:
$$\kappa_{\partial_v} = \langle D_{\partial_v}\partial_v,\partial_r\rangle|_{\mathcal{H}_S} = -\frac12\partial_r\langle \partial_v,\partial_v\rangle |_{\mathcal{H}_S}= -\frac12\partial_r(1-\frac{2M}{r})|_{\mathcal{H}_S}=\frac{1}{4M}.$$
Since $\kappa_{\partial_v}$ is positive, it follows also that $\mathcal{H}_S$ is strictly stable.\\
\indent  The standard spherically symmetric, past-directed Null Cone of the Schwarzschild spacetime is given by $\Omega_S:=\{v=v_0\}$. Our Killing vector field highlights the fact that our choice of $v_0$ is isometrically indistinguishable from any other. We also find the null geodesic generator for $\Omega_S$ from $\text{grad}(v) = \partial_r$. We conclude that $r$ restricts to an affine parameter along the geodesics generating $\Omega_S$ and therefore any cross section $\Sigma$ can be given as a graph over $\mathbb{S}^2$, $\omega(\vartheta,\varphi) = r|_\Sigma$. 
\begin{lemma}\label{l7} (see \cite{R})
Given a cross section $\Sigma:=\{r=\omega\}$ of $\Omega_S$ we have: 
\begin{align*}
\gamma= \omega^2\mathring{\gamma},\qquad\ubar\chi=\frac{1}{\omega}\gamma,&\qquad\tr\ubar\chi=\frac{2}{\omega},\qquad\zeta=-\s{d}\log\omega,\\
\chi=\frac{1}{\omega}(1-\frac{2M}{\omega}+|\nabla\omega|^2)\gamma-2\nabla^2\omega,&\qquad\tr\chi=\frac{2}{\omega}\Big(1-\frac{2M}{\omega}-\omega^2\Delta\log\omega\Big),\\
\rho&=\frac{2M}{\omega^3}.
\end{align*}
where $\mathring\gamma$ is the round metric on $\mathbb{S}^2$.
\end{lemma}
As a result of the strong Maximum Principle, we observe that the only cross-section of $\Omega_S$ satisfying $\tr\chi=0$ is given by $\Sigma_0:=\{v=v_0,r=2M\}=\Omega_S\cap\mathcal{H}_S$. Moreover, since $\tau = \zeta-d\log\tr\ubar\chi = 0$, $\Sigma_0$ is a quasi-round MOTS (in-fact $\Sigma_0$ is `round') since $\mathcal{K} =\mathcal{K}+\nabla\cdot\tau= \frac{1}{4M^2}$. Translating $\Sigma_0=\mathcal{H}_S\cap\Omega_S$ along $\partial_v$, we have identified the unique foliation of $\mathcal{H}_S$ by quasi-round MOTS.
\begin{remark}\label{r2}
We take this opportunity to further motivate the mass functional, (\ref{e2}). From Lemma \ref{l7}, any cross-section $\Sigma:=\{r=\omega\}\subset\Omega_S$ exhibits the mass:
$$m(\Sigma) = \frac12\Big(\int_\Sigma\Big(\frac{2M}{\omega^3}\Big)^{\frac23}\frac{\omega^2d\mathring{A}}{4\pi}\Big)^{\frac32}=M.$$
We observe, regardless of the cross-section of $\Omega_S$, the mass functional is only sensitive to the Trautman-Bondi mass $M$. This is the main advantage over prior attempts at studying the Null Penrose Inequality, as this mass is generally insensitive to infinitesimal spacetime boosts unlike the Hawking Energy and energy functionals generally (see \cite{R}). We will exploit this advantage in Theorem \ref{t4}.
\end{remark}
\subsection{Assumptions and Asymptotics}
We wish to show that small metric perturbations of the Schwarzschild spacetime satisfies the Null Penrose Inequality. In order to do so, we will need certain asymptotics along a Null Cone to explicitly identify the Trautman-Bondi Energy and Mass.\\ 
\indent The analysis of ODEs in this section will make frequent use of the following result:
\begin{lemma}\label{l8}(\cite{hale2009ordinary}, Corollary 6.3)\\
Let $w(t,u)$ be continuous on $a\leq t<b$, $u\geq0$ with the initial value problem  $\dot{u} = w(t,u)$ having a unique solution $u(t)\geq 0$ for $a\leq t<b$. If $f:[a,b)\times \mathbb{R}^n\to\mathbb{R}^n$ is continuous and 
$$|f(t,x)|\leq w(t,|x|),\,\,a\leq t<b,\,\,x\in\mathbb{R}^n,$$
then the solutions of 
$$\dot{x} = f(t,x),\,\,|x(a)|\leq u(a)$$
exist on $[a,b)$ and $|x(t)|\leq u(t)$.
\end{lemma}

Upon applying Theorem \ref{t3} to the strictly stable quasi-round MOTS $\Sigma_0$ in $\Omega_S$, we will need to assume the existence of some $\epsilon>0$, so that any $\lambda\leq \epsilon$ ensures the existence of a Null Cone $\Omega_\lambda$ to the past of the quasi-round MOTS $\Sigma_\lambda$ along $L_\lambda^-\in\Gamma(T^\perp\Sigma_\lambda)$. We denote by $\ubar L_\lambda\in\Gamma(T^\perp\Omega_\lambda)\subset\Gamma(T\Omega_\lambda)$ the extension of $L_\lambda^-$ by $D_{\ubar L_\lambda}\ubar L_\lambda=0$. For convenience, we will take an associated geodesic foliation $\{\Sigma_s\}_s$ such that $s_0 = 1$. The following definition follows along the lines of \cite{MS1}:
\begin{definition}\label{d13} Denoting a smooth $k$-tensor by $T|_p:T_p\Omega_\lambda\otimes_1...\otimes_{k-1}T_p\Omega_\lambda\to\mathbb{R}$, then for a basis extension $\{X_i\}\subset E(\Sigma_\lambda)$ along $\ubar L_\lambda$:
\begin{enumerate}
\item We say $T$ is transversal whenever $T(\ubar L_\lambda,X_{i_1},...,X_{i_{k-1}})=...=T(X_{i_1},...,X_{i_{k-1}},\ubar L_\lambda) = 0$
\item We say $T=O_n(s^{-m})$ whenever
$$s^m(\pounds_{X_{i_1}}\cdots\pounds_{X_{i_j}}T(s))(X_{l_1},\cdots,X_{l_k})=O(1),\,\,(0\leq j\leq n)$$
\item Given $\sum_i|X_i(1)|_\lambda^2\leq C$ for some fixed constant $C$ and all $\lambda\leq \epsilon$, we say $T(\lambda,s)=O^\lambda_n(s^{-m})$ whenever $T=O_n(s^{-m})$ and
$$\lim\sup_{\lambda\to0}\Big(\sup_{\Omega_\lambda}|s^m(\pounds_{X_{i_1}}\cdots\pounds_{X_{i_j}}T(\lambda,s))(X_{l_1},\cdots,X_{l_k})|\Big)=0,\,\,(0\leq j\leq n).$$
\end{enumerate}
\end{definition}
With this definition in hand, and some $0<\delta<1$, we assume the following conditions on the ambient curvature:
\begin{align}
\ubar \alpha_\lambda(X_i,X_j)&:=\langle R_{\ubar L_\lambda X_i}\ubar L_\lambda,X_j\rangle= O^\lambda_4(s^{-1-\delta})\label{e18}\\
\mathcal{L}_{\ubar L_\lambda}\ubar\alpha_\lambda&=O^{\lambda}(s^{-2-\delta})\label{e19}\\
G_{\ubar L_\lambda}(X_i)&:=G(\ubar L_\lambda,X_i) = O^\lambda_3(s^{-2-\delta})\label{e20}\\
G(L_\lambda, \ubar L_\lambda) &= O^\lambda_2(s^{-3-\delta})\label{e21}.
\end{align}

From assumptions (\ref{e18}-\ref{e21}) we will be able to make sense of the notion of total mass for $\Omega_\lambda$. In order to do so we need the following known result which, to the author's understanding, is due to S. Alexakis \cite{A}. We provide a proof for completeness and context regarding later results.
\begin{proposition}\label{p7}({\cite{A}}, Lemma 4.1)
For sufficiently small $\epsilon$ there exists a function $\theta=O^\lambda_4(1)$, and transverse 2-tensors $\bar\gamma,\hat{\ubar\chi}=O^\lambda_4(1)$ on $\Omega_\lambda$ such that 
$$\gamma(\lambda,s) = s^2\gamma(\lambda,1)+s^2\bar\gamma,\,\,\tr\ubar\chi = \frac{2}{s}+\frac{\theta}{s^2},\,\,\pounds_{\ubar L_\lambda}{\gamma}=\hat{\ubar\chi}+\frac12\tr\ubar\chi\gamma.$$
Moreover, any $T(\lambda,s)\in\{\theta,\bar\gamma,\hat{\ubar\chi}\}$ observed under the natural pull-back to $\Sigma_\lambda$, has a four times continuously differentiable limit $\displaystyle{T_\infty(\lambda) = \lim_{s\to\infty}T(\lambda,s)}$ whereby 
$$\pounds_{x_{i_1}}\cdots\pounds_{x_{i_j}}T_\infty(\lambda) = \lim_{s\to\infty}\pounds_{X_{i_1}}\cdots\pounds_{X_{i_j}}T(\lambda,s),$$ 
for $x_{i} = X_{i}(1)$, $\{X_i\}\subset E(\Sigma_\lambda)$, $0\leq j\leq 4$. 
\end{proposition}
\begin{proof}
We start by Lie-dragging $\gamma(\lambda,1)$ along $\ubar L_\lambda$ to the rest of $\Omega_\lambda$, denoted by $\gamma_0$. We also define $\tilde\gamma:=\frac{1}{s^2}\gamma$. Taking a basis extension $\{X_i\}\subset E(U)$, $U\subset\Sigma_\lambda$, we then observe from the structure equations, Proposition \ref{p3}:
\begin{align*}
\frac{d}{ds}{\Big(\sqrt{\frac{\det\tilde\gamma}{\det\gamma_0}}\Big)}^{-1}&=-\frac{\sqrt{\det\gamma_0}}{\det\tilde\gamma}(-\frac{2}{s^3}\sqrt{\det\gamma}+\frac{1}{s^2}\tr\ubar\chi\sqrt{\det\gamma})\\
&=-\frac{\theta}{s^2}\Big(\sqrt{\frac{\det\tilde\gamma}{\det\gamma_0}}\Big)^{-1}\\
\frac{d}{ds}{\tilde{\gamma}}_{ij}&=-\frac{2}{s^3}\gamma_{ij}+\frac{2}{s^2}\hat{\ubar\chi}_{ij}+\frac{1}{s^2}\tr\ubar\chi\gamma_{ij}\\
&=\frac{\theta}{s^2}\tilde\gamma_{ij}+\frac{2}{s^2}\hat{\ubar\chi}_{ij}\\
\frac{d}{ds}{\theta}&=2s\tr\ubar\chi-4+s^2(\frac{2}{s^2}-\frac12\tr\ubar\chi^2-|\hat{\ubar\chi}|^2-G(\ubar L_\lambda,\ubar L_\lambda))\\
&=-\frac{1}{2s^2}(s^4\tr\ubar\chi^2-4s^3\tr\ubar\chi+4)-s^2\gamma^{ij}\gamma^{kl}\hat{\ubar\chi}_{ik}\hat{\ubar\chi}_{jl}-s^2\gamma^{ij}\ubar\alpha_{ij}\\
&=-\frac{1}{2s^2}\theta^2-\frac{1}{s^2}\tilde\gamma^{ij}\tilde\gamma^{kl}\hat{\ubar\chi}_{ik}\hat{\ubar\chi}_{jl}-\tilde\gamma^{ij}\ubar\alpha_{ij}
\end{align*}
\begin{align*}\frac{d}{ds}{\hat{\ubar\chi}}_{ij}&=-\hat{\ubar\alpha}_{ij}+2\gamma^{kl}\hat{\ubar\chi}_{ik}\hat{\ubar\chi}_{lj}\\
&=-\hat{\ubar\alpha}_{ij}+\frac{2}{s^2}\tilde\gamma^{kl}\hat{\ubar\chi}_{ik}\hat{\ubar\chi}_{lj}.
\end{align*}
We now define $\bar{\gamma}_{ij}:=\tilde\gamma_{ij}-{\gamma_0}_{ij}$, $\mathcal{D}:=(\sqrt{\frac{\det\tilde\gamma}{\det\gamma_0}})^{-1}-1$, and $u^2(\lambda,s):=\mathcal{D}^2+\theta^2+\sum_{ij}\Big((\bar{\gamma}_{ij})^2+\hat{\ubar\chi}_{ij}^2\Big)$. We therefore have
\begin{align*}
|\frac{d}{ds}{\mathcal{D}}|&\leq\frac{u}{s^2}(u+1)\\
\sqrt{\sum_{ij}{(\frac{d}{ds}{\bar{\gamma}}_{ij}})^2}&=\sqrt{\sum_{ij}{{(\frac{d}{ds}{\tilde\gamma}_{ij}}})^2}\leq \frac{u}{s^2}(u+|\gamma_0|)+\frac{2}{s^2}u\\
|\frac{d}{ds}{\theta}|&\leq\frac{1}{2s^2}u^2+\frac{2}{s^2}|\tilde\gamma^{-1}|^2|\hat{\ubar\chi}|^2+|\tilde\gamma^{-1}|\frac{c(\lambda)}{s^{1+\delta}}\\
&=\frac{1}{2s^2}u^2+\frac{2}{s^2}\det\tilde\gamma^{-2}|\tilde\gamma|^2|\hat{\ubar\chi}|^2+\det\tilde\gamma^{-1}|\tilde\gamma|\frac{c(\lambda)}{s^{1+\delta}}\\
&\leq\frac{1}{2s^2}u^2+\frac{(u+1)^2}{\det\gamma_0}(u+|\gamma_0|)\Big(\frac{2u^2}{s^2}\frac{(u+1)^2}{\det\gamma_0}(u+|\gamma_0|)+\frac{c(\lambda)}{s^{1+\delta}}\Big)\\
\sqrt{\sum_{ij}{(\frac{d}{ds}{\hat{\ubar\chi}})^2_{ij}}}&\leq|\ubar\alpha-\frac12\tilde\gamma^{ij}\ubar\alpha_{ij}\tilde\gamma|+\frac{2}{s^2}\det\tilde\gamma^{-1}|\tilde\gamma||\hat{\ubar\chi}|^2\\
&\leq\frac{c(\lambda)}{s^{1+\delta}}(1+\det\tilde\gamma^{-1}|\tilde\gamma|^2)+2\frac{u^2}{s^2}\frac{(u+1)^2}{\det\gamma_0}(u+|\gamma_0|)\\
&\leq\frac{c(\lambda)}{s^{1+\delta}}+\frac{(u+1)^2}{\det\gamma_0}(u+|\gamma_0|)\Big(\frac{c(\lambda)}{s^{1+\delta}}(u+|\gamma_0|)+2\frac{u^2}{s^2}\Big)\end{align*}
for some continuous function $c:[0,\epsilon]\to[0,\infty)$ whereby $c(0)=0$. After a simple modification of $c(\lambda)$ we may therefore conclude that
$$\sqrt{(\frac{d}{ds}{\mathcal{D}})^2+(\frac{d}{ds}{\theta})^2+\sum_{ij}({\frac{d}{ds}{\bar{\gamma}}_{ij}})^2+\sum_{ij}(\frac{d}{ds}{\hat{\ubar\chi}}_{ij})^2)}\leq \frac{uP(u)+c(\lambda)}{s^{1+\delta}}$$
for some seventh order polynomial $P$ with constant positive coefficients.\\
\indent Regarding solutions to the ODE 
$$\frac{d}{ds}{y}=\frac{y P(y)+c(\lambda)}{s^{1+\delta}},\,\,y(\lambda,1)=\sup_{\Sigma_\lambda}|\hat{\ubar\chi}|(1),$$
immediately we note that $y(\lambda,s)$ is monotone increasing in $s$. For any constant $\alpha>0$ we also claim an $\epsilon(\alpha)>0$ such that $\displaystyle{\lim_{s\to\infty}y(\lambda,s)\leq \alpha}$ for all $\lambda\leq\epsilon(\alpha)$. Otherwise, there exists sequences $\{\lambda_i\},\{s_i\}$ such that $\displaystyle{\lim_{i\to\infty}\lambda_i=0},\lim_{i\to\infty}s_i=\infty$, and $y(\lambda_i,s_i)>\alpha$. We obtain a contradiction using the inequality
$$\int_{y(\lambda_i,1)}^\alpha\frac{du}{uP(\alpha)+c(\lambda_i)}\leq\int_{y(\lambda_i,1)}^{y(\lambda_i,s_i)}\frac{du}{uP(u)+c(\lambda_i)}=\int_1^{s_i}\frac{1}{t^{1+\delta}}dt\leq \frac{1}{\delta},$$
since $y(\lambda_i,1) = \displaystyle{\sup_{\Sigma_{\lambda_i}}|\hat{\ubar\chi}|(1)}\xrightarrow{\lambda_i\to0} 0$ causes the first integral to blow-up.\\
In-fact, for 
$$F(x,\lambda):=\int_{\displaystyle{\sup_{\Sigma_\lambda}|\hat{\ubar\chi}|}}^x\frac{du}{uP(u)+c(\lambda)},$$
we see $F_x(x,\lambda)>0$ is continuous for all $(x,\lambda)\in (0,\infty)\times(0,\epsilon)$. Therefore, a standard generalization of the Implicit Function Theorem (see, for example, \cite{loomis2014advanced} Theorem 9.3) ensures that $\displaystyle{y(\lambda,\infty):=\lim_{s\to\infty}y(\lambda,s)}$, given implicitly via $F(y(\lambda,\infty),\lambda) = \frac{1}{\delta}$, is both unique and continuous in $\lambda$. With a similar blow-up argument as above, we also conclude that $\displaystyle{\lim_{\lambda\to 0^+}y(\lambda,\infty)=0}$. From Lemma \ref{l8} it follows that,
$$\sqrt{\mathcal{D}^2+\theta^2+\sum_{ij}\Big({\bar{\gamma}_{ij}}^2+\hat{\ubar\chi}_{ij}^2\Big)}=u(\lambda,s) \leq y(\lambda,s)=O^\lambda(1).$$
Denoting $\vec{T} = (\mathcal{D},\bar{\gamma}_{ij}, \theta, \hat{\ubar\chi}_{ij})$, if we integrate our propagation equations using the bound on $u(\lambda,s)$ we observe
$$\sup_{\Sigma_\lambda}|\vec{T}(s_m)-\vec{T}(s_n)|\leq c(\lambda)|\frac{1}{s_m^\delta}-\frac{1}{s_n^\delta}|.$$
From this Cauchy sequence we deduce uniform convergence to a continuous limit $\displaystyle{\vec{T}_\infty = \lim_{s\to\infty}\vec{T}(s)}$. We now take a derivatives of our propagation equations, using the bound on $u(\lambda,s)$, to observe for any $T\in\{\mathcal{D},\bar{\gamma},\theta,\hat{\ubar\chi}\}$,
$$\frac{d}{ds}(\pounds_{X_k}{T})_{ij} = (\pounds_{X_k}\frac{d}{ds}{{T}})_{ij} = A_i^l(\pounds_{X_k}{T})_{lj}+{h}_{ij}$$
whereby $A_i^j,h_{ij}=\frac{1}{s^{1+\delta}}O^\lambda(1)$. From this, similarly as before, $| \frac{d}{ds}\vec{\pounds_{X_i}T}|\leq\frac{c(\lambda)}{s^{1+\delta}}(|\vec{\pounds_{X_i}T}|+1)$ and it follows by Lemma \ref{l8} that 
$$|\vec{\pounds_{X_i}{T}}(s)|\leq \Big(1+\sqrt{\sum_{jk}(\pounds_{X_i}\hat{\ubar\chi})^2_{jk}(\lambda,1)}\Big)e^{c(\lambda)(1-\frac{1}{s^{\delta}})}-1=O^{\lambda}(1).$$
Once again we integrate this linear system to conclude
$$\sup_{\Sigma_\lambda}|\vec{\pounds_{X_i}T}(s_m)-\vec{\pounds_{X_i}T}(s_n)|\leq c(\lambda)|\frac{1}{s_m^\delta}-\frac{1}{s_n^\delta}|,$$
and we have uniform convergence to a limit $\displaystyle{\lim_{s\to\infty}}\vec{\pounds_{X_i}T} = \vec{T}^\infty_i$. With an adapted co-ordinate system $(x^i,s)$ for $\Omega_\lambda$, $(x^i)$ local coordinates on $\Sigma_\lambda$, we observe $\pounds_{v}((\pi_s)_\star T(s)) = (\pi_s)_\star(\pounds_{V}T(s))$, where $\pi_s:=\pi|_{\Sigma_s}$ is a diffeomorphism, $V\in E(\Sigma_\lambda)$, $v = V(0)$. We conclude, $\pounds_{x_i}T_\infty = T^\infty_i$. The remaining limits now follow analogously from established decay on lower derivatives, for up to three additional derivatives.
\end{proof}

From Proposition \ref{p7} we are now in a position to show that (\ref{e18}-\ref{e21}) are independent of our choice of geodesic generator $\ubar L$. First we recall that any geodesic generator is given by $\ubar L_a = a\ubar L_\lambda$ for some $a\in\mathcal{F}(\Sigma_\lambda)$ (and of-course we imply $a$ is Lie-dragged throughout $\Omega_\lambda$), with affine parameters related by $s-1 = a (s_a-1)$. More generally, if we don't require that our affine parameters restrict to constants on $\Sigma_\lambda$, or equivalently that our initial cross-section $\{s_a=1\}\neq\Sigma_\lambda$, we have $s = as_a +\phi$ for some $a,\phi\in\mathcal{F}(\Sigma_\lambda)$. Moreover, given any extension basis $\{X_i\}$ associated to $\ubar L$, we have another, $\{X_i^a\}$ associated to $\ubar L_a$, given by $X_i^a = X_i+(s_aX_ia+X_i\phi)\ubar L_\lambda$. We see therefore that (\ref{e18}) translates, $\langle R_{\ubar L_a X^a_i}\ubar L_a,X_j^a\rangle = a^2\langle R_{\ubar L_\lambda X_i}\ubar L_\lambda,X_j\rangle = a^2O^\lambda_4(s^{-1-\delta}) = O^\lambda_4(s_a^{-1-\delta})$. Similarly we also observe (\ref{e19}) since $\ubar\alpha_\lambda$ is a transversal tensor. To observe (\ref{e20}) we use Proposition \ref{p7}, $G(\ubar L_\lambda,\ubar L_\lambda) = \gamma^{ij}(\ubar\alpha_\lambda)_{ij} = s^{-2}\tilde\gamma^{ij}(\ubar\alpha_\lambda)_{ij} = O^\lambda_4(s^{-3-\delta})$, so we have $G(\ubar L_a,X^a_i) = aG(\ubar L_\lambda,X_i)+a(s_aX_i(a)+X_i(\phi))G(\ubar L_\lambda,\ubar L_\lambda) = O^\lambda_3(s^{-2-\delta})+s_aO^\lambda_4(s^{-3-\delta}) = O_3^\lambda(s_a^{-2-\delta})$. For (\ref{e21}), we have 
\begin{align*}
G(L_a,\ubar L_a) &= G\Big(\frac{1}{a}(L_\lambda+|\nabla(s_aa+\phi)|^2\ubar L_\lambda-2\nabla(s_aa+\phi)),a\ubar L_\lambda\Big)\\
&= G(L_\lambda,\ubar L_\lambda)+s^{-2}\tilde\gamma^{ij}(s_aa_i+\phi_i)(s_aa_j+\phi_j)G(\ubar L_\lambda,\ubar L_\lambda)-2s^{-2}(s_aa_i+\phi_i)\tilde\gamma^{ij}G(\ubar L_\lambda,X^a_j)\\
&=O^\lambda_2(s^{-3-\delta})+O^\lambda_4(s^{-3-\delta})+O^\lambda_3(s^{-3-\delta})\\
&=O^\lambda_2(s_a^{-3-\delta}).
\end{align*}
As a result of (\ref{e18},\ref{e19}), Proposition \ref{p7} holds for any geodesic foliation off of $\Sigma_\lambda$. One can even generalize, given $a\in\mathcal{F}(\Sigma_\lambda)$, by assigning $\phi(t) = o(t)$ along a foliation defined by $\Sigma_t:=\{s=at+\phi(t)\}$, called an \textit{asymptotically geodesic} foliation of $\Omega_\lambda$. However, for our purposes, geodesic foliations will suffice.\\
\indent For our next result we need the known fact:
\begin{proposition}(\cite{R}, Theorem 3.2)\label{p8}\\
Assume $\{\Sigma_s\}_s$ is a foliation of a Null Cone $\Omega$ along the flow vector $\ubar L = \sigma L^-$ (as described at the beginning of Section 2), then
\begin{align*}
\frac{d}{ds}{\rho}+\frac32\sigma\rho&=\frac{\sigma}{2}\Big(\frac12\langle\vec{H},\vec{H}\rangle\Big(|\hat{\chi}^-|^2+G(L^-,L^-)\Big)+|\tau|^2-\frac12G(L^-,L^+)\Big)\\
&\quad+\Delta\Big(\sigma(|\hat{\chi}^-|^2+G(L^-,L^-)\Big)-2\nabla\cdot(\sigma\hat{\chi}^-\circ\tau)+\nabla\cdot(\sigma G_{L^-}).
\end{align*}
\end{proposition}
\begin{proposition}\label{p9}
For sufficiently small $\epsilon$, along the foliation $\{\Sigma_s\}_s$ associated to $\ubar L_\lambda$, we conclude $\tau = O^{\lambda}_3(s^{-1})$, and $\{\Sigma_s\}_s$ is a doubly convex foliation of $\Omega_\lambda$.
Moreover, there exists functions $\bar\rho,\bar{H}^2=O^\lambda_2(1)$ on $\Omega_\lambda$ such that 
$$\rho_s = \frac{1}{s^3}\Big(\frac{4\pi}{|\Sigma_\lambda|}+\bar\rho\Big),\,\,\langle\vec{H},\vec{H}\rangle_s=\frac{1}{s^2}\Big(\frac{16\pi}{|\Sigma_\lambda|}(1-\frac{1}{s})+\bar{H}^2\Big).$$
For $T\in\{\bar\rho,\bar{H}^2\}$, we also have a twice continuously differentiable limit $\displaystyle{T_\infty(\lambda)=\lim_{s\to\infty}T(\lambda,s)}$, whereby $\pounds_{x_{i_1}}\cdots\pounds_{x_{i_j}}T_\infty(\lambda) = \displaystyle{\lim_{s\to\infty}\mathcal{L}_{X_{i_1}}\cdots\mathcal{L}_{X_{i_j}}T(\lambda,s)}$, $0\leq j\leq 2$, for $x_i = X_i(1)$, $\{X_i\}\subset E(U)$, $U\subset\Sigma_\lambda$.
\end{proposition}
\begin{proof}
From the structure equations and Proposition \ref{p7}, we see that
\begin{align*}
\frac{d}{ds}{(\sqrt{\frac{\det\gamma}{\det\gamma_0}}\tau_i)}&=\sqrt{\frac{\det\gamma}{\det\gamma_0}}\Big(-\nabla\cdot\hat{\ubar\chi}_i+G(\ubar L_\lambda, X_i) +X_i\frac{|\hat{\ubar\chi}|^2+G(\ubar L_\lambda,\ubar L_\lambda)}{\tr\ubar\chi}\Big)\\
&=O^{\lambda}_3(1).
\end{align*}
Therefore, we have $\tau_i = \sqrt{\frac{\det\gamma_0}{\det\gamma}}\tau_i(\lambda,1)+\sqrt{\frac{\det\gamma_0}{\det\gamma}}\int_1^sO^{\lambda}_3(1)dt = O^{\lambda}_3(s^{-1})$. Using this, we apply Proposition \ref{p8} to $\tilde\rho:=s^3\rho$, and Proposition \ref{p6} to $\tilde{H}^2:=s^2\langle\vec{H},\vec{H}\rangle$ (using $\psi = \tr\ubar\chi$ in Proposition \ref{p6}):
\begin{align*}
\frac{d}{ds}{\tilde\rho}&=-\frac{3\theta}{2s^2}\tilde\rho+\frac14\tilde{H}^2\frac{s}{\tr\ubar\chi}\Big(|\hat{\ubar\chi}|^2+G(\ubar L_\lambda,\ubar L_\lambda)\Big)+\frac12(s^3\tr\ubar\chi)|\tau|^2-\frac14(s^3\tr\ubar\chi)G(\ubar L_\lambda, L_\lambda)\\
&\qquad\qquad+s^3\Delta\frac{|\hat{\ubar\chi}|^2+G(\ubar L_\lambda,\ubar L_\lambda)}{\tr\ubar\chi}+s^3\nabla\cdot G_{\ubar L_\lambda}-2s^3\nabla\cdot({\hat{\ubar\chi}}\circ\tau)\\
&=-\frac{3\theta}{2s^2}\tilde\rho+\frac14\tilde{H}^2\frac{s}{\tr\ubar\chi}\Big(|\hat{\ubar\chi}|^2+G(\ubar L_\lambda,\ubar L_\lambda)\Big) + O^{\lambda}_2(s^{-1-\delta})\\
\frac{d}{ds}{\tilde{H}}^2&=\frac{2}{s}\tilde{H}^2-\frac32\tr\ubar\chi\tilde{H}^2-\tilde{H}^2\frac{|\hat{\ubar\chi}|^2+G(\ubar L_\lambda,\ubar L_\lambda)}{\tr\ubar\chi}+\frac12\tr\ubar\chi\tilde{H}^2+\frac{2}{s}\tr\ubar\chi\tilde\rho\\
&\qquad\qquad-2s^2\Delta\tr\ubar\chi-4s^2\nabla\cdot(\tr\ubar\chi\tau)-2s^2\tr\ubar\chi|\tau|^2+(s^2\tr\ubar\chi)G(\ubar L_\lambda,L_\lambda)\\
&=\frac{2}{s}\tr\ubar\chi\tilde\rho-\Big(\frac{\theta}{s^2}+\frac{|\hat{\ubar\chi}|^2+G(\ubar L_\lambda, \ubar L_\lambda)}{\tr\ubar\chi}\Big)\tilde{H}^2+O^{\lambda}_2(s^{-2}).
\end{align*}
Denoting $\vec{T} = (\tilde\rho, \tilde{H}^2)$ we see from the propagation equations that
$$\frac{d}{ds}{\vec{T}} = A\vec{T}+\vec{h}$$
where $A,h = \frac{1}{s^{1+\delta}}O(1)$. From this we conclude that $|\frac{d}{ds}{\vec{T}}|\leq \frac{C}{s^{1+\delta}}(|\vec{T}|+1)$ and therefore Lemma \ref{l8} gives
$$|\vec{T}|\leq(1+ \frac{4\pi}{|\Sigma_\lambda|})e^{C(1-\frac{1}{s^\delta})}-1=O(1).$$
With this bound, we return to the propagation equations to find
$$\frac{d}{ds}{\tilde\rho}=-\frac{3\theta}{2s^2}\tilde\rho+O^{\lambda}(s^{-1-\delta}),$$
from which we deduce, similarly as in Proposition \ref{p7}, a continuous limit $\displaystyle{\rho_\infty:=\lim_{s\to\infty}\tilde{\rho}}$. Moreover, defining $\bar{\rho}:=\tilde\rho-\frac{4\pi}{|\Sigma_\lambda|}$, we have $|\frac{d}{ds}{\bar\rho}|\leq \frac{c(\lambda)}{s^{1+\delta}}(|\bar\rho|+1)$ ensuring that $|\bar\rho|\leq e^{c(\lambda)(1-\frac{1}{s^\delta})}-1=O^{\lambda}(1)$. Similarly, we verify $\bar{H}^2:=\tilde{H}^2- \frac{16\pi}{|\Sigma_\lambda|}(1-\frac{1}{s})=O^{\lambda}(1)$, with a continuous limit as $s\to\infty$. Consider now the propagation of $(\bar{\rho},\bar{H}^2)$:
$$\frac{d}{ds}\begin{pmatrix}
{\bar\rho}\\
{\bar{H}}^2
\end{pmatrix}
=\begin{pmatrix}
O^{\lambda}_4(s^{-2})&O^{\lambda}_4(s^{-1-\delta})\\
O_4(s^{-2})&O^{\lambda}_4(s^{-2})
\end{pmatrix}
\begin{pmatrix}
\bar\rho\\
\bar{H}^2
\end{pmatrix}
+\begin{pmatrix}
O^{\lambda}_2(s^{-1-\delta})\\
O^{\lambda}_2(s^{-2})
\end{pmatrix}.
$$
Taking a derivative results in a linear system of the form $\frac{d}{ds}{X_i}\vec{T} = A({X_i}\vec{T})+\vec{h}$ whereby $A,h=\frac{1}{s^{1+\delta}}O(1)$, and we deduce boundedness of $|{X_i}\vec{T}|$. From this we can bootstrap as in Proposition \ref{p7} to deduce continuous limits and control in $\lambda$ for up to one additional derivative.\\
\indent Finally, we return to the equation
$$\frac{d}{ds}{\tilde\rho} = -\frac{3\theta}{2s^2}\tilde\rho+O^{\lambda}_2(s^{-1-\delta}),$$
and we conclude that 
\begin{align*}
\tilde\rho &= \frac{4\pi}{|\Sigma_\lambda|}e^{\int_1^s\frac{\theta}{t^2}dt}+\int_1^sO^{\lambda}_2(t^{-1-\delta})dt\\
X_i\tilde\rho&=\frac{4\pi}{|\Sigma_\lambda|}\Big(\int_1^s\frac{X_i\theta}{t^2}dt\Big)e^{\int_1^s\frac{\theta}{t^2}dt}+\int_1^sO^{\lambda}_1(t^{-1-\delta})dt = \int_1^sO^\lambda(t^{-1-\delta})dt\\
(\nabla_s^2\tilde\rho)(X_i,X_j)&=\frac{4\pi}{|\Sigma_\lambda|}\Big(\int_1^s\frac{X_i\theta}{t^2}dt\int_1^s\frac{X_j\theta}{t^2}dt+\int_1^s\frac{(\nabla_t^2\theta)(X_i,X_j)}{t^2}dt\Big)e^{\int_1^s\frac{\theta}{t^2}dt}+\int_1^sO^{\lambda}(t^{-1-\delta})dt\\
&=\int_1^sO^\lambda(t^{-1-\delta})dt.
\end{align*}
From the first expression we conclude that sufficiently small $\epsilon$ will ensure $\tilde\rho>0$, from the second and third, 
$$|\Delta_s\log\rho_s|\leq |\gamma^{ij}\frac{(\nabla^2\tilde\rho)_{ij}}{\tilde\rho}|+|\gamma^{ij}\frac{X_i\tilde\rho X_j\tilde\rho}{\tilde\rho^2}|\leq \frac{c(\lambda)}{s^2}(1-\frac{1}{s^{\delta}})$$
for some continuous $c(\lambda)$ such that $c(0) = 0$. We leave to the reader the simple exercise of verifying $\langle\vec{H},\vec{H}\rangle_s\geq \frac{16\pi}{s^2|\Sigma_\lambda|}(1-\frac{1}{s})-\frac{c_1(\lambda)}{s^2}(1-\frac{1}{s^{\delta}})$. As a result, for some modified $c(\lambda)$,
$$\frac14\langle\vec{H},\vec{H}\rangle_s-\frac13\Delta_s\log\rho_s\geq\frac{1}{s^2}\Big( \frac{4\pi}{|\Sigma_\lambda|}(1-\frac{1}{s})-c(\lambda)(1-\frac{1}{s^{\delta}})\Big)=:\frac{1}{s^2}f(\lambda,s).$$
From $\partial_sf(\lambda,s) = \frac{1}{s^2}\Big(\frac{4\pi}{|\Sigma_\lambda|}-\delta c(\lambda)s^{1-\delta}\Big)$, we observe precisely one critical value at $s=\Big(\frac{4\pi}{\delta c(\lambda)|\Sigma_\lambda|}\Big)^{\frac{1}{1-\delta}}.$ For sufficiently small $\epsilon$, we have $\frac{4\pi}{|\Sigma_\lambda|}\geq c(\lambda)$, $\lambda\leq \epsilon$, giving $f(\lambda,1) = 0$, $(\partial_sf)(\lambda,1)=\frac{4\pi}{|\Sigma_\lambda|}-\delta c(\lambda)>0$, and $\displaystyle{\lim_{s\to\infty}f(\lambda,s)} = \frac{4\pi}{|\Sigma_\lambda|}-c(\lambda)\geq 0$. We conclude therefore that $f(\lambda,s)\geq 0$ for all $\lambda\leq\epsilon$, and $\{\Sigma_s\}_s$ is a doubly convex foliation.
\end{proof}
\subsection{Trautman-Bondi Energy and Mass}
\begin{definition}\label{d14} Given a spacelike 2-sphere $\Sigma$ with mean curvature $\vec{H}=\tr_\Sigma\II$, the \textit{Hawking Energy} is given by
$$E_H(\Sigma):=\sqrt{\frac{|\Sigma|}{16\pi}}\Big(1-\frac{1}{16\pi}\int_\Sigma\langle\vec{H},\vec{H}\rangle dA\Big).$$
\end{definition}

\indent We will be interested in taking a limit of the Hawking Energy along foliations to `null infinity' of our Null Cone. Existence of a limit has been analyzed in work of Christodoulou-Klainermann \cite{christodoulou2014global}, generalized by Bieri \cite{bieri2010}, also Klainerman-Nicol\`{o} \cite{klainerman2003evolution}, Chru\'{s}ciel-Paetz \cite{chrusciel2014mass}, and a setting similar to ours due to Mars-Soria \cite{MS1}. From this limit, one is able to define the notion of the total mass of a null cone $\Omega$ called the \textit{Trautman-Bondi mass} (see Definition \ref{d15}). Historically, the definition of mass at null infinity can be traced back to Trautman \cite{Trautman:1958zz}, generalizing, and predating, a more explicit coordinate based construction by Bondi et al. \cite{bondi1962gravitational,sachs1962gravitational}.  We refer the reader to \cite{bieri2016future} for more details, including how the Trautman-Bondi mass relates to the famous one of Arnowitt-Deser-Misner \cite{arnowitt2008republication} at spacelike infinity.\\
\indent The Schwarzschild spacetime and standard Null Cone $\Omega_S$ is a convenient toy model to motivate the Trautman-Bondi notion of total Energy and Mass. First we recall, given the standard round metric $\mathring{\gamma}$ on $\mathbb{S}^2$, any conformal rescaling $\gamma_\omega:=\omega^2\mathring\gamma$ exhibits Gauss curvature
$$\mathcal{K}_\omega = \frac{1}{\omega^2}(1-\mathring\Delta\log\omega).$$
Up to a diffeomorphism, the \textit{Uniformization Theorem} states that any metric on $\mathbb{S}^2$ can be given as such a conformal rescaling of $\mathring\gamma$. We therefore refer to $S = (\mathbb{S}^2,\gamma)$ as a \textit{round sphere} whenever the a metric $\gamma$ induces constant Gauss curvature, $\mathcal{K}=\frac{4\pi}{|S|}$ via the Gauss-Bonnet theorem. It follows that $(\mathbb{S}^2,\gamma_\omega)$ is a round sphere, if and only if $\omega$ solves the non-linear equation 
$$1-\Big(\frac{\omega}{r_0}\Big)^2 = \mathring\Delta\log\Big(\frac{\omega}{r_0}\Big)$$
for $r_0:=\sqrt{\frac{|S|}{4\pi}}$. Solutions take the form $\omega(\vartheta,\varphi) = r_0\frac{\sqrt{1-|\vec{v}|^2}}{1-\vec{v}\cdot\vec{n}(\vartheta,\varphi)}$, for some $\vec{v}$ inside the unit ball $\mathring{B}^3\subset\mathbb{R}^3$, and $\vec{n}(\vartheta,\varphi)$ the unit position vector. \\
\indent For cross-sections $\Sigma:=\{r=\omega\}\subset\Omega_S$, we observe that $\Sigma = (\mathbb{S}^2,\gamma_\omega)$. Using Lemma \ref{l7}, a simple calculation gives that any round sphere $\Sigma_{\vec{v}}^{r_0}\hookrightarrow\Omega_S$ observes a Hawking Energy $E_H(r_0,\vec{v}) = \frac{M}{\sqrt{1-|\vec{v}|^2}}$ which is precisely the observed (special relativistic) energy of a particle of mass $M$ traveling at velocity $\vec{v}$ relative to its observer. In this setting the Trautman-Bondi energy is given by 
$$E_{TB}(\vec{v}) = \displaystyle{\lim_{r_0\to\infty}}E_H(\Sigma^{r_0}_{\vec{v}}) = \frac{M}{\sqrt{1-|\vec{v}|^2}}.$$ 
It follows that the energy of an asymptotically round foliation approaches the mass $M$ only if $\vec{v}=0$ `at infinity' as expected. Thus giving the Trautman-Bondi mass as $m_{TB}(\Omega_S) = \inf_{\vec{v}}E_{TB}(\vec{v}) = M$.\\\\
\indent Given Proposition \ref{p7}, we are in a position to define the total Trautman-Bondi energy and mass of our Null Cones $\Omega_\lambda$. We will temporarily denote by $E^a(\Sigma_\lambda)$ the the set of vector field extensions off of $\Sigma_\lambda$ along $\ubar L_a:=a\ubar L_\lambda$, whereby $s-1=a(s_a-1)$, and $0<a\in\mathcal{F}(\Sigma_\lambda)$. We recall $X\in E(\Sigma_\lambda)$ produces $X^a:=X+X(a)(s_a-1)\ubar L_\lambda\in E^a(\Sigma_\lambda)$. Relative to a basis extension $\{X_i\}\subset E(U)$, $U\subset\Sigma_\lambda$, we observe from Proposition \ref{p7} that the sphere `at infinity' inherits the metric $\displaystyle{\gamma^\infty_{ij}:=\lim_{s\to\infty}\frac{\gamma_{ij}(\lambda,s)}{s^2}=\lim_{s\to\infty}\frac{1}{s^2}\langle X_i,X_j\rangle}$ along $\{\Sigma_s\}_s$. Along $\{\Sigma_{s_a}\}$, the sphere at infinity inherits the metric 
$$\gamma^a_{ij}=\displaystyle{\lim_{s_a\to\infty}\frac{1}{s_a^2}\langle X^a_i,X^a_j\rangle = \lim_{s_a\to\infty}\frac{\gamma_{ij}(\lambda,s)}{s^2}\frac{s^2}{s_a^2}=a^2\gamma^\infty_{ij}}.$$
By the Uniformization Theorem, we may therefore choose $a_\lambda$ such that $a_\lambda^2\gamma^\infty = \ring{\gamma}$. We will denote
$$\phi_{\vec{v}}(\vartheta,\varphi):=\frac{\sqrt{1-|\vec{v}|^2}}{1-\vec{v}\cdot\vec{n}(\vartheta,\varphi)}$$
where $\vec{v}\in\mathring{B}^3\subset\mathbb{R}^3$, and $\vec{n}(\vartheta,\varphi)\in\partial\mathring{B}^3$.
\begin{definition}\label{d15}
The total Trautman-Bondi Energy $E_{TB}(\lambda, \vec{v})$ of $\Omega_\lambda$ is given by
$$E_{TB}(\lambda,\vec{v}):=\lim_{t\to\infty}E_H(\Sigma_t)$$
whereby $s-1=(a_\lambda\omega_{\vec{v}})(t-1)$. The total Trautman-Bondi Mass $m_{TB}(\lambda)$ is given by
$$m_{TB}(\lambda) = \inf_{\{\vec{v}||\vec{v}|<1\}}E_{TB}(\lambda,\vec{v}).$$
\end{definition}
\subsection{Stability of the Null Penrose Inequality}
We may re-write the Hawking Energy for a cross-section $\Sigma$ using the Gauss-Bonnet and Divergence theorems as
$$E_H(\Sigma) = \frac{1}{4\pi}\sqrt{\frac{|\Sigma|}{16\pi}}\int_\Sigma\rho dA.$$
From Propositions \ref{p1}, \ref{p7}, and \ref{p9} we observe that sufficiently small $\epsilon$ ensures that 
$$\sqrt{\frac{|\Sigma_\lambda|}{16\pi}}\leq \lim_{s\to\infty}m(\Sigma_s) = \frac12\Big(\int_{\mathbb{S}^2}(\rho_\infty)^{\frac23}\frac{dA_\infty}{4\pi}\Big)^{\frac32}$$
where $\rho_\infty:=\displaystyle{\lim_{s\to\infty}}(s^3\rho_s)$, $dA_\infty:=\displaystyle{\lim_{s\to\infty}}\frac{1}{s^2}\sqrt{\frac{\det\gamma_s}{\det\gamma_0}}dA_0$. Therefore, our final result follows as soon as we show
$$\frac12\Big(\int_{\mathbb{S}^2}(\rho_\infty)^{\frac23}\frac{dA_\infty}{4\pi}\Big)^{\frac32} \leq m_{TB}(\lambda).$$
In order to do so, we need the following Proposition:
\begin{proposition}(\cite{R}, Theorem 4.1)\label{p10}
Consider a cross-section $\Sigma_\omega:=\{s=\omega\}\subset\Omega_\lambda$, $\omega\in\mathcal{F}(\Sigma_\lambda)$, with associated flux function $\rho$. Then, the point-wise decomposition of $\rho$ relative to data $(\rho_s,\ubar\chi_s,\tau_s)$ associated to the background foliation $\{\Sigma_s\}_s$ along $\ubar L_\lambda$ is given by:
\begin{align*}
\rho &= \Big(\rho_s +\frac{|\hat{\ubar\chi}_s|^2+G(\ubar L_\lambda, \ubar L_\lambda)}{\tr\ubar\chi_s}\Big(\Delta_s\omega-2\hat{\ubar\chi}_s(\nabla_s\omega,\nabla_s\omega)\Big)+2\nabla_s\omega\frac{|\hat{\ubar\chi}_s|^2+G(\ubar L_\lambda, \ubar L_\lambda)}{\tr\ubar\chi_s}\\
&\qquad+\frac12\Big(|\hat{\ubar\chi}_s|^2+G(\ubar L_\lambda,\ubar L_\lambda)+2\ubar L_\lambda\frac{|\hat{\ubar\chi}_s|^2+G(\ubar L_\lambda, \ubar L_\lambda)}{\tr\ubar\chi_s}\Big)|\nabla_s\omega|^2+G(\ubar L_\lambda,\nabla_s\omega)-2\hat{\ubar\chi}_s(\vec{\tau}_s,\nabla_s\omega)\Big)\circ\pi.
\end{align*}
\end{proposition}
\begin{theorem}\label{t4}

Consider the Schwarzschild spacetime metric in ingoing Eddington-Finkelstein coordinates $(v,r,\vartheta,\varphi)$:
$$g_0 = -(1-\frac{2M}{r})dv\otimes dv+(dr\otimes dv+dv\otimes dr)+r^2\big(d\vartheta\otimes d\vartheta+(\sin\vartheta)^2d\varphi\otimes d\varphi\big),$$
on the neighborhood $\mathcal{U}:=(-\epsilon_0+v_0,\epsilon_0+v_0)\times(r_0,\infty)\times\mathbb{S}^2$, $\epsilon_0,r_0>0$. Consider also a smooth path of metrics, $\lambda\to g_\lambda\in \text{Sym}(T^\star\mathcal{U}\otimes T^\star\mathcal{U})$, $0\leq \lambda \leq c$, satisfying the Dominant Energy Condition. For $\{\Sigma_\lambda\}_{0\leq \lambda\leq \epsilon}$ the corresponding family of smooth quasi-round MOTS of Theorem \ref{t3}, whereby $\Sigma_0=\Omega_S\cap\mathcal{H}_S\cong\{v_0\}\times\{2M\}\times\mathbb{S}^2$ is a standard quasi-round MOTS of Schwarzschild, we assume the existence of an $\epsilon_1\leq \epsilon$ such that the past directed Null Cones, $\Omega_\lambda\supset\Sigma_\lambda$, $0\leq \lambda\leq \epsilon_1$, exist satisfying the curvature decay conditions (\ref{e18})-(\ref{e21}). Then, there exists an $0<\epsilon_2\leq \epsilon_1$ such that the Null Penrose Inequality
$$\sqrt{\frac{|\Sigma_\lambda|}{16\pi}}\leq m_{TB}(\lambda)$$
holds for $0\leq \lambda\leq \epsilon_2$.

\end{theorem}
\begin{proof}
For any geodesic foliation $\{\Sigma^a_t\}\subset\Omega_\lambda$ whereby $\omega^a_t:=s|_{\Sigma_t}=a(t-1)+1$, $a\in\mathcal{F}(\Sigma_\lambda)$, we observe from Propositions \ref{p7}, \ref{p9}, and \ref{p10}
$$\rho_t=\frac{(s^3\rho_s)|_{\Sigma_t}}{\omega_t^3}+O^{\lambda}(t^{-3-\delta})\implies\lim_{t\to\infty}t^3\rho_t = \frac{\rho_\infty}{a^3}\geq 0.$$
The Hawking Energy along $\{\Sigma^a_t\}$ satisfies
\begin{align*}
\lim_{t\to\infty}E_H(\Sigma^a_t) &=\lim_{t\to\infty}\Big(\frac{1}{4\pi}\sqrt{\frac{|\Sigma^a_t|}{16\pi t^2}}\int_{\mathbb{S}^2}t^3\rho_t\frac{(\omega^a_t)^2}{t^2}\Big(\frac{1}{s^2}\sqrt{\frac{\det\gamma_s}{\det\gamma_0}}\Big)\Big|_{\Sigma_t^a}dA_0\Big)\\
& = \frac{1}{4\pi}\sqrt{\frac{1}{16\pi}\int_{\mathbb{S}^2} a^2dA_\infty}\int_{\mathbb{S}^2}\frac{\rho_\infty}{a}dA_\infty.
\end{align*}
It's a simple exercise using H\"{o}lder's inequality and the fact that $\rho_\infty\geq 0$ to show
$$\Big(\int_{\mathbb{S}^2}(\rho_\infty)^{\frac23}dA_\infty\Big)^{\frac32}\leq\inf_{a>0}\Big\{\sqrt{\int_{\mathbb{S}^2}a^2dA_\infty}\int_{\mathbb{S}^2}\frac{\rho_\infty}{a}dA_\infty\Big\}.$$
In-fact, we obtain equality by taking $a_\varepsilon:=(\rho_\infty+\varepsilon)^{\frac13}$, $\varepsilon>0$, and noting
\begin{align*}
\Big(\int_{\mathbb{S}^2}(\rho_\infty)^{\frac23}dA_\infty\Big)^{\frac32} &= \lim_{\varepsilon\to 0}\Big\{ \sqrt{\int_{\mathbb{S}^2}a_{\varepsilon}^2dA_\infty}\int_{\mathbb{S}^2}\frac{\rho_\infty+\varepsilon}{a_{\varepsilon}}dA_\infty\Big\}\\
&=\lim_{\varepsilon\to 0} \Big\{ \sqrt{\int_{\mathbb{S}^2}a_{\varepsilon}^2dA_\infty}\int_{\mathbb{S}^2}\frac{\rho_\infty}{a_{\varepsilon}}dA_\infty\Big\}\\
&\geq \inf_{a>0}\Big\{\sqrt{\int_{\mathbb{S}^2}a^2dA_\infty}\int_{\mathbb{S}^2}\frac{\rho_\infty}{a}dA_\infty\Big\}.
\end{align*}
Using the Uniformization theorem, every Trautman-Bondi energy can be realized as the limit of the Hawking Energy along some geodesic foliation off of $\Sigma_\lambda$ (see, for example \cite{MS1}). As a result,
$$\frac12\Big(\int_{\mathbb{S}^2}(\rho_\infty)^{\frac23}\frac{dA_\infty}{4\pi}\Big)^{\frac32}\leq \inf_{\{\vec{v}||\vec{v}<1\}}E_{TB}(\lambda,\vec{v}) = m_{TB}(\lambda).$$
\end{proof}
\section*{Acknowledgments}
This material is based upon work supported by the National Science Foundation under Award No. 1703184. The author would like to thank Hubert Bray and Richard Schoen for their continued support and encouragement, as well as Marc Mars, and Chao Li for helpful conversations on topics addressed by this paper. This work developed out of the author's visit to the Erwin Schr\"{o}dinger International Institute for Mathematics and Physics (ESI) during the ``Geometry and Relativity Conference" in the summer of 2017, which the author would like to acknowledge.
\bibliographystyle{abbrv}
\normalbaselines 
\bibliography{horizon stability.bbl} 
\end{document}